\documentclass[reqno,a4paper]{article}
\usepackage{hyperref}
\usepackage{amsmath}
\usepackage[english]{babel}
\usepackage{latexsym}
\usepackage{amssymb}
\usepackage{amscd}
\usepackage{amsthm,amsgen,amstext,amsbsy,amsopn,dsfont}
\usepackage{math rsfs}

\usepackage{epsfig,graphicx,graphics}
\usepackage[latin1]{inputenc}
\usepackage{xspace}
\usepackage{amsxtra}
\usepackage{bm}

\usepackage{color}

\newcommand{\bdm}{\begin{displaymath}}
\newcommand{\edm}{\end{displaymath}}
\newcommand{\bdn}{\begin{eqnarray}}
\newcommand{\edn}{\end{eqnarray}}
\newcommand{\bay}{\begin{array}{c}}
\newcommand{\eay}{\end{array}}
\newcommand{\ben}{\begin{enumerate}}
\newcommand{\een}{\end{enumerate}}
\newcommand{\beq}{\begin{equation}}
\newcommand{\eeq}{\end{equation}}

\newcommand{\chiin}{\chi_{\mathrm{in}}}
\newcommand{\chiout}{\chi_{\mathrm{out}}}

\newcommand{\R}{\mathbb{R}}
\newcommand{\N}{\mathbb{N}}

\newcommand{\C}{\mathbb{C}}

\newcommand{\F}{\mathcal{F}}
\newcommand{\E}{\mathcal{E}}
\newcommand{\A}{\mathcal{A}}
\newcommand{\B}{\mathcal{B}}

\newcommand{\I}{\mathcal{I}}

\newcommand{\LL}{\mathcal{L}}

\newcommand{\PP}{\mathbb{P}}

\newcommand{\x}{\mathbf{x}}

\newcommand{\one}{{\ensuremath {\mathds 1} }}

\newcommand{\ep}{\varepsilon}

\newcommand{\Om}{\Omega}
\newcommand{\om}{\omega}

\newcommand{\dd}{\partial}
\newcommand{\half}{\frac{1}{2}}

\newcommand{\supp}{\mathrm{supp}}
\newcommand{\wto}{\rightharpoonup}

\newcommand{\intR}{\int_{\R ^2}}
\newcommand{\intRN}{\int_{\R ^{2N}}}

\newcommand{\bral}{\left<}

\newcommand{\ketr}{\right>}

\newcommand{\rhoP}{\rho_{\Psi}}

\newcommand{\pot}{V_{\om,k}}
\newcommand{\potN}{V_{\om,k}^N}

\newcommand{\LLLN}{\mathfrak{H} ^{N}}
\newcommand{\LLL}{\mathfrak{H}}
\newcommand{\LLLh}{H ^{\rm L}}
\newcommand{\LLLf}{\E ^{\rm L}}
\newcommand{\LLLe}{E ^{\rm L}}
\newcommand{\LLLm}{\Psi ^{\rm L}}

\newcommand{\Lgv}{L_{\rm qh}}

\newcommand{\Barg}{\B}
\newcommand{\BargN}{\B ^{N}}
\newcommand{\Bargh}{H ^{\B}}
\newcommand{\Bargm}{F^ {\B}}
\newcommand{\Barge}{E ^{\B}}

\newcommand{\Lau}{\Psi ^{\rm Lau}}
\newcommand{\cLau}{c _{\rm Lau}}
\newcommand{\intN}{\I_N}

\newcommand{\Ker}{\mathrm{Ker} (\intN)}

\newcommand{\PKerp}{P_{\mathrm{Ker} (\intN) ^{\perp}}}

\newcommand{\gap}{\mathrm{gap}}
\newcommand{\PsiGV}{\Psi ^{\rm qh}}

\newcommand{\mopt}{m_{\rm opt}}

\newcommand{\ropt}{r_{\rm opt}}

\newcommand{\muN}{\mu_{N}}
\newcommand{\muNone}{\mu_{N} ^{(1)}}
\newcommand{\ZN}{\mathcal{Z}_{N}}
\newcommand{\AN}{\A_{N}}

\newcommand{\FN}{\F_{N}}
\newcommand{\FNe}{F_{N}}
\newcommand{\HN}{H_{N}}
\newcommand{\Vm}{W_m}

\newcommand{\FNt}{\tilde{\F}_{N}}
\newcommand{\FNet}{\tilde{F}_{N}}
\newcommand{\HNt}{\tilde{H}_{N}}

\newcommand{\MFf}{\E ^{\rm MF}}
\newcommand{\MFe}{E ^{\rm MF}}
\newcommand{\rhoMF}{\varrho ^{\rm MF}}
\newcommand{\hMF}{h_{\varrho ^{\rm MF}}}

\newcommand{\MFfel}{\E ^{\rm el}}
\newcommand{\MFeel}{E ^{\rm el}}
\newcommand{\rhoMFel}{\varrho ^{\rm el}}
\newcommand{\hMFel}{h_{\varrho ^{\rm el}}}

\newcommand{\MFfth}{\E ^{\rm th}}
\newcommand{\MFeth}{E ^{\rm th}}
\newcommand{\rhoMFth}{\varrho ^{\rm th}}

\newcommand{\Zth}{Z^{\rm th}}

\newcommand{\MFftel}{\tilde{\E} ^{\rm el}}
\newcommand{\MFetel}{\tilde{E} ^{\rm el}}
\newcommand{\rhoMFtel}{\tilde{\varrho} ^{\rm el}}

\newcommand{\MFftth}{\tilde{\E} ^{\rm th}}
\newcommand{\MFetth}{\tilde{E} ^{\rm th}}
\newcommand{\rhoMFtth}{\tilde{\varrho} ^{\rm th}}
\newcommand{\Ztht}{\tilde{Z}^{\rm th}}

\newcommand{\rhoMFF}{\hat{\varrho} ^{\rm MF}}

\newcommand{\rhoMFelF}{\hat{\varrho} ^{\rm el}}

\newcommand{\Rminus}{R_m ^-}
\newcommand{\Rplus}{R_m ^+}

\newcommand{\muxi}{\mu_{x_i}}
\newcommand{\muxj}{\mu_{x_j}}
\newcommand{\delxi}{\delta_{x_i}}

\newtheorem{teo}{Theorem}[section]
\newtheorem{lem}{Lemma}[section]
\newtheorem{pro}{Proposition}[section]

\newcounter{remark}[section]
\newenvironment{rem}{\stepcounter{remark} \vspace{0,1cm} \noindent \textit{Remark \thesection.\theremark}\,}{\vspace{0,2cm}}

\numberwithin{equation}{section}

\begin{document} 
\title{Quantum Hall phases and plasma analogy in rotating trapped Bose gases}
\author{N. Rougerie ${}^{1}$, S. Serfaty ${}^{2,3}$, J. Yngvason ${}^{4,5}$\\
\normalsize\it ${}^{1}$ Universit\'e Grenoble 1 and CNRS, LPMMC, UMR 5493, BP 166, 38042 Grenoble, France.\\
\normalsize\it ${}^{2}$ Universtit\'e Paris 6,  Laboratoire Jacques-Louis Lions, UMR 7598, Paris, F-75005 France.\\
\normalsize\it ${}^{3}$ Courant Institute, New York University, 251 Mercer st, NY 10012, USA.\\
\normalsize\it ${}^{4}$ Fakult\"at f\"ur Physik, Universit{\"a}t Wien, Boltzmanngasse 5, 1090 Vienna, Austria.\\
\normalsize\it ${}^{5}$ Erwin Schr{\"o}dinger Institute for Mathematical Physics, Boltzmanngasse 9, 1090 Vienna, Austria.
}

\date{April 15th, 2013}
 
\maketitle
\centerline{
{\it Dedicated to Herbert Spohn on the occasion of his retirement from the TU M\"unchen}}

\begin{abstract}
A bosonic analogue of the fractional quantum Hall effect occurs in rapidly rotating trapped Bose gases: There is a transition from uncorrelated Hartree states to strongly correlated states such as the Laughlin wave function. This physics may be described by effective Hamiltonians with delta interactions acting on a bosonic $N$-body Bargmann space of analytic functions. In a previous paper \cite{RSY} we studied the case of a quadratic plus quartic trapping potential and derived conditions on the parameters of the model for its ground state to be asymptotically strongly correlated. This relied essentially on energy upper bounds using quantum Hall trial states, incorporating the correlations of the Bose-Laughlin state in addition to a multiply quantized vortex pinned at the origin. In this paper we investigate in more details the density of these trial states, thereby substantiating further the physical picture described in \cite{RSY}, improving our energy estimates and allowing to consider more general 
trapping potentials. Our analysis is based on the interpretation of the densities of quantum Hall trial states as Gibbs measures of classical 2D Coulomb gases (plasma analogy). New estimates on the mean-field limit of such systems are presented. 
\end{abstract}

\tableofcontents

\section{Introduction}\label{sec:intro}

The advent in the 90's of powerful techniques to cool and trap atoms has opened the way to new investigations of quantum phenomena on a macroscopic scale. It has now become possible to isolate extremely cold and dilute atomic gases and maintain them in metastable states that can be efficiently modeled as ground states of effective many-body Hamiltonians with repulsive short range interactions. Although the atoms typically trapped are neutral, one can impose an artificial magnetic field to them. This is usually achieved by rotating the trap \cite{Co,Fet} although there now exist more refined techniques \cite{DGJO}. Exploiting the analogy between the Lorentz and Coriolis force, one can easily realize that the Hamiltonian for neutral atoms in a rotating frame resembles that of charged particles in a uniform magnetic field.

The analogy has been demonstrated by the nucleation of quantized vortices in cold rotating Bose gases (see \cite{Fet} and references therein, in particular \cite{BSSD}), similar to those appearing in type II superconductors submitted to external magnetic fields. In this regime, the atoms all condense in the same one-particle state and the gas forms a Bose-Einstein condensate (BEC). The peculiar properties of quantized vortices, demonstrating the superfluid nature of BECs, have motivated numerous theoretical and mathematical works, see \cite{Aft,Fet,Co,CPRY1,CPRY2} and references therein. 

An even more striking possibility is to create with cold atomic gases phases characteristic of the fractional quantum Hall effect (FQHE), a phenomenon originally observed in 2D electron gases submitted to very large magnetic fields \cite{STG}. In this regime, the Bose gas is no longer a condensate and mean-field theories fail to capture the physics of the system. One has to use a truly many-body description, and strongly correlated phases, such as the celebrated Laughlin state \cite{Lau} may occur. This regime has been hitherto elusive in trapped Bose gases. For large rotation speeds, the centrifugal force can compensate the trapping force and lead to an instability of the gas. As of now, the competing demands of a large rotation speed on the one hand, necessary for entering the FQHE regime, and of a sufficiently stable system on the other hand, have left the FQHE regime unattainable with current technology (see \cite{RRD} for a precise discussion of this point).

Several ideas to get around this difficulty have been proposed, and it is one of those that we  examine from a mathematical point of view, here and in the companion paper \cite{RSY}. As we will see, the proposed modification of the experimental set-up also leads to new physics and mathematics.

\medskip

When discussing the quantum Hall regime it is common \cite{Aft,Fet,Co} to assume that the gas is essentially 2D, that the single-particle states are restricted to lie in the ground eigenspace of the ``magnetic'' kinetic energy operator (i.e. in the lowest Landau level) and that the inter-particles interactions are described by a Dirac delta potential. We will give more details on these approximations, that can to some extent be backed with rigorous mathematics \cite{LS}, in Section~\ref{sec:model}. Assuming their validity, the system may be described in the rotating frame by the following many-body Hamiltonian \footnote{$\rm L$ stands for Landau, referring to the lowest Landau level.} :
\begin{equation}\label{eq:intro LLLh}
\LLLh =  \sum_{j=1} ^N \left(V (z_j) - \frac12 \Om ^2 |z_j| ^2 \right) + g \sum_{i<j} \delta (z_i-z_j)
\end{equation}
``acting'' on the Lowest Landau Level (LLL) for $N$ bosons
\begin{equation}\label{eq:intro LLLN}
\LLLN = \left\{ F(z_1,\ldots,z_N) e^{-\sum_{j=1} ^N |z_j| ^2 / 2 } \in L ^2 (\R ^{2N}),\: F \mbox{ holomorphic and symmetric }\right\}. 
\end{equation}
We have here identified the positions of the particles in the plane with complex numbers $z_1,\ldots,z_N$ and by `symmetric' we mean invariant under the exchange of $z_i$ and $z_j$ for $i\neq j$. The parameter $\Om \geq 0$ is the angular frequency, with the convention that the rotation vector points in the direction perpendicular to the plane in which the particles are confined. The interaction between particles $i$ and $j$ is given by the delta function $g \delta (z_i-z_j)$ where $g$ is proportional to the scattering length of the original 3D interaction potential between the particles \cite{LS}. In our units $\hbar = m = 1$ and we will also fix $\Omega= 1$ in the sequel.

To find the ground state energy $\LLLe$ of the system we minimize the expectation value $\bral \LLLh \ketr_{\Psi}$ amongst $N$-particles states $\Psi$ in $\LLLN$:
\begin{equation}\label{eq:intro LLLe}
\LLLe := \inf\left\{\LLLf[\Psi] = \bral \Psi, \LLLh \: \Psi \ketr,\: \Psi \in \LLLN, \Vert \Psi \Vert_{L^2 (\R ^{2N})} = 1  \right\}.
\end{equation}
More precisely, the energy of a state $\Psi$ reads 
\begin{equation}\label{eq:intro LLLf}
\LLLf[\Psi] :=  \sum_{j=1} ^N \intRN \left(V (z_j) - \frac12 \Om ^2 |z_j| ^2 \right) |\Psi (Z)| ^2 dZ + g \sum_{1 \leq i<j \leq N} \int_{\R^{2(N-1)}} \left| \Psi (z_i = z_j) \right| ^2 d\hat{Z}_i
\end{equation}
where $Z= (z_1,\ldots,z_N)$, $dZ = dz_1\ldots dz_N$, $\Psi (z_i = z_j)$ is the function $\Psi$ evaluated on the diagonal $z_i = z_j$ and $d\hat{Z}_i$ is $dZ$ whith the factor $dz_i$ removed. Note that it makes perfect sense to use a delta interaction potential in $\LLLN$ : since functions in this space are all smooth, the integral
\[
\bral \Psi, \delta(z_1 - z_2) \Psi \ketr := \int_{\R^{2(N-1)}} \left| \Psi (z_2,z_2,z_3,\ldots,z_N) \right| ^2 dz_2 \ldots dz_N = \int_{\R^{2(N-1)}} \left| \Psi (z_1 = z_2) \right| ^2 d\hat{Z}_1
\]
is always well-defined. We remark that strictly speaking \eqref{eq:intro LLLh} does not act on \eqref{eq:intro LLLN} since it does not leave that space invariant. The actual Hamiltonian under consideration is that arising from the quadratic form \eqref{eq:intro LLLf}. We will also use another equivalent formulation of the problem below.

\medskip

In most experiments, the trapping potential $V$ is quadratic 
$$V (x)=  \frac12 \Om _{\perp} ^2  |x| ^2$$
and hence the \emph{effective potential}, 
\begin{equation}\label{eq:intro pot eff}
V_{\rm eff} (x) =   V(x) - \frac12 \Om ^2 |x| ^2,
\end{equation}
that includes the effect of the centrifugal force is bounded below only if $\Om \leq \Om_{\perp}$. There is thus a maximum angular frequency that one can impose on the gas, and the FQHE regime of strong correlations is expected to occur when $\Om \to \Om_{\perp}$, i.e. when the centrifugal force almost compensates the trapping force. To avoid this singularity of the limit $\Om \to \Om_{\perp}$, it has been proposed \cite{MF,Vie} to add a weak anharmonic component to the trapping potential. A simple example is provided by a quadratic plus quartic trap of the form
\begin{equation}\label{eq:intro pot}
V (r) = \frac12 \Om_{\perp} ^2 r ^2 + k r ^4. 
\end{equation}
With such a potential there is no theoretical limit to the rotation speed that one can impose on the system. One may thus expect the regime $\Om \to \Om_{\perp}$ to be less singular and more manageable experimentally. This has been demonstrated in the Gross-Pitaevskii regime by using this kind of trap to rotate Bose-Einstein condensates beyond the centrifugal limit \cite{BSSD}.

\medskip

In \cite{RSY} and the present paper we focus on the particular case where the trap $V$ in \eqref{eq:intro LLLh} is given by \eqref{eq:intro pot} :
\begin{equation}\label{eq:LLLh}
\LLLh =  \sum_{j=1} ^N \left(\om |z_j| ^2 + k |z_j| ^4 \right) + g \sum_{i<j} \delta (z_i-z_j)
\end{equation}
with 
$$\om := \frac12 \left(\Om_{\perp} ^2 - \Om ^2\right).$$
We are mainly interested in two questions : 
\begin{enumerate}
\item In which parameter regime can one obtain strongly correlated states as approximate ground states for this model ?
\item Can the addition of the quartic part of the potential lead to new physics and to phases not accessible with the purely harmonic trap ?
\end{enumerate}
Let us now be more precise about what we mean by strongly correlated states. In the simpler case where $k=0$, it is easy to see (we give more details in Section \ref{sec:model}) that for sufficiently large ratio $g/ \om$, the exact ground state of the system is given by the bosonic Laughlin state
\begin{equation}\label{eq:intro Laughlin}
\Lau = \cLau \prod_{1\leq i < j \leq N} \left( z_i -z_j\right) ^2 e ^{-\sum_{j=1} ^N |z_j| ^2 / 2 } 
\end{equation}
where $\cLau$ is a normalization constant. This wave function\footnote{or rather its fermionic analogue where the exponent of the $(z_i-z_j)$ factor is an odd number, ensuring fermionic symmetry} was originally introduced in \cite{Lau,Lau2} as a proposal to approximate the ground state of a 2D electron gas in a strong magnetic field. The correlations encoded in the holomorphic factor in \eqref{eq:intro Laughlin} decrease the interaction energy and even cancel it in our case of contact interactions.  

More generally (see \cite[Section 2.2]{LS} for further discussion), it is expected that for $g/(N\om)$ of order unity, one encounters a series of strongly correlated states with smaller and smaller interaction energy, ultimately leading to the Laughlin state. Several candidates have been proposed in the literature (see \cite{Co,Vie} for review) for the other correlated states that should occur before eventually reaching the Laughlin state. Although there is some numerical evidence that they have a good overlap with the true ground states, none of these trial functions is as firmly established as the Laughlin state.

On the other hand, for small $g/(N\om)$ one can show that the ground state energy of \eqref{eq:LLLh} is  well approximated by Gross-Pitaevskii theory \cite{LSY}, that is by taking a Hartree trial state 
\begin{equation}\label{eq:intro Hartree}
\Psi (z_1,\ldots,z_N)= \prod_{j=1} ^N \phi (z_j) 
\end{equation}
for some single particle wave function $\phi$. The corresponding GP theory with states restricted to the lowest Landau level has some very interesting features, studied in \cite{ABN1,ABN2}. One thus goes from a fully uncorrelated state for small $g/(N\om)$ to a highly correlated Laughlin state for large $g/(N\om)$.

\medskip

Now, what changes in this scenario when a quartic component is added to the trap? The tools of \cite{LSY} still apply in the Gross-Pitaevskii regime \cite{Dra}, which leads to a GP theory with some specific new aspects \cite{BR,R}. The strongly correlated regime is more difficult: The Laughlin state is no longer an exact eigenstate of the Hamiltonian, so it is not obvious that one should eventually reach this state. One should certainly expect that the ground state will essentially live in the kernel of the interaction operator $\intN = \sum_{i<j} \delta(z_i-z_j)$
\begin{equation}\label{eq:intro kernel}
\Ker = \left\{  \Lau F(z_1,\ldots,z_N),\: F \mbox{ holomorphic and symmetric}  \right\} \subset \LLLN
\end{equation}
in a regime where the interaction energy dominates the physics, i.e. when $g$ is sufficiently large and $\om,k$ sufficiently small. It is of interest to derive explicit conditions on the order of magnitude of these parameters that is needed to obtain a strongly correlated state in a certain limit, and also to obtain rigorous estimates for the energy $\LLLe$ in this regime. 

Also, a remarkable new feature of the model with $k\neq 0$ is that it is now possible to consider negative values of $\om$, that is, to go beyond the centrifugal limit of the harmonic trap by taking $\Om > \Om_{\perp}$. In this case, the effective potential \eqref{eq:intro pot eff} has a local maximum at the origin, which can lead to new physics. Indeed, it can be shown (see below) that the matter density of the Laughlin state is almost flat in a disc around the origin. If the trap develops a maximum there, the Laughlin state will have a large potential energy and one may expect that other strongly correlated phases belonging to $\Ker$, with depleted density at the center of the trap, will be preferred to the Laughlin state, even for large $g$ and small $\om,k$. 

\medskip

Our main results on the model \eqref{eq:LLLh}, derived in \cite{RSY}, can be summarized as follows 
\begin{itemize}
\item we rigorously identify conditions on the parameters of the problem $g,\om,k$ and $N$ under which the ground state of $\LLLh$ is fully correlated in the sense that its projection on the orthogonal complement of \eqref{eq:intro kernel} vanishes in a certain limit.
\item within this regime we obtain upper and lower bounds to the energy whose order of magnitudes match in the appropriate limit.
\item we can identify a regime where the Laughlin state does \emph{not} approximate the true ground state, even if it almost fully lives in $\Ker$. We prove that a state with higher angular momentum is preferred, which can never happen in a purely harmonic trap with $k=0$.
\end{itemize}
Precise statements are given in Section \ref{sec:model results} below, after we have recalled several facts about the formulation of Problem \eqref{eq:intro LLLe} using the Bargmann space of holomorphic functions. The proofs are given in \cite[Section 5]{RSY}. They mostly rely on adequate energy upper bounds derived with fully-correlated states of the form \eqref{eq:intro kernel}. Indeed, since the Laughlin state is no longer a true eigenstate of the Hamiltonian, the evaluation of its potential energy, and that of other candidate trial states, is a non trivial problem. Also, as discussed above, the effective potential changes from having a local minimum at the origin to having a local maximum when $\om$ is decreased. We thus need to have some flexibility in the matter density of our trial states to adapt to this behavior, which leads us to the form 
\begin{equation}\label{eq:trial states}
\PsiGV _m = c_m \prod_{j=1} ^N z_j ^m \prod_{1 \leq i<j\leq N} \left( z_i - z_j\right) ^2 e^{-\sum_{j=1} ^N |z_j| ^2 / 2}
\end{equation}
were $c_m$ is a normalization constant. One recovers the pure Laughlin state for $m=0$ and the states with $m>0$ are usually referred to as Laughlin quasi-holes \cite{Lau,Lau2} (hence the label $\rm qh$). They were introduced as elementary excitations of the pure Laughlin state. The factor $\prod_{j=1} ^N z_j ^m$ is interpreted as an additional multiply quantized vortex located at the origin. Its role is to deplete the density of the state when $\omega < 0$ to reduce potential energy.

We have no rigorous argument allowing to prove that the true ground states are asymptotically of the form above in the strongly correlated regime although we do believe that it is the case. A rudimentary lower bound to the energy confirms however that the trial states \eqref{eq:trial states} at least give the correct order of magnitude for the energy in the strongly correlated regime, as stated in Theorem \ref{teo:result ener} below.

The method we use in \cite{RSY} to evaluate the energy of these states relies on the representation of the effective potential 
\begin{equation}\label{eq:defi pot}
\pot (r) = \om r^2 + k r^4 
\end{equation}
in terms of angular momentum operators on the Bargmann space of analytic functions (see Section \ref{sec:model} below).  This leads relatively easily to energy estimates that imply criteria for full correlation.  On the other hand, this method allows only a limited physical interpretation of the results because it says nothing about the character of the particle density and its change as the parameters are varied.
Moreover, it is from the outset limited to radial potentials of a special kind.  Anisotropic potentials, but also a radial potential like
\[
V (r) = - k_4 r ^4 + k_6 r ^6 
\] 
with $k_4,k_6 >0$ cannot be treated by this method. For these reasons the focus of the present paper is on a different aspect than in \cite{RSY}, namely on {\it rigorous estimates on the particle density in strongly correlated states}. Such density estimates  imply in particular an alternative method for obtaining energy estimates that works also for more complicated potentials than \eqref{eq:defi pot}. In the following we summarize this method.

Since we are interested in the strongly correlated regime, all the trial states we use will belong to $\Ker$ and thus have zero interaction energy. Our task is then to calculate their potential energy 
\begin{equation}\label{eq:intro pot ener}
N \int_{\R ^2} V (z) \rhoP (z) dz 
\end{equation}
where
\begin{equation}\label{eq:intro density}
\rhoP (z) = \int_{\R ^{2(N-1)}} \left| \Psi (z,z_2,\ldots,z_N)\right| ^2 dz_2\ldots dz_N
\end{equation}
is the one-particle density of the state $\Psi$, that we have defined so that its integral is $1$. Given a candidate trial state we thus want to evaluate precisely what the corresponding matter density is. This will also provide a better understanding of the wave functions \eqref{eq:trial states} that play a central role in FQHE physics \cite{Gir,LFS,STG}. As we will see, this understanding is crucial for the interpretation of our results on the minimization problem \eqref{eq:intro LLLe}.

Our main tool is the well-known plasma analogy, originating in \cite{Lau,Lau2}, wherein the density of the Laughlin state is interpreted as the Gibbs measure of a classical 2D Coulomb gas (one component plasma). More precisely, after a scaling of space variables, one can identify the $N$-particle density of the Laughlin state with the Gibbs measure of a 2D jellium with mean-field scaling, that is a system of $N$ particles in the plane interacting via weak (with a prefactor $N^{-1}$) logarithmic pair-potentials and with a constant neutralizing background. Within this analogy the  vortex of degree $m$ in \eqref{eq:trial states} is interpreted as an additional point charge pinned at the origin.

Existing knowledge (e.g. \cite{CLMP,Kie1,KS,MS,SS}) about the mean-field limit for classical particles then suggests that we should be able to extract valuable information on the density of the Laughlin and related states in the limit $N\to \infty$. For our purpose, precise estimates that are not available in the literature are required, so we develop a new strategy for the study of the mean field limit that gives explicit and quantitative estimates on the fluctuations about the mean field density that are of independent interest.

%Due to the central role that the Laughlin state plays in the understanding of the FQHE, determining with a good precision the matter density corresponding to states of the form \eqref{eq:trial states} is an interesting problem in its own right \cite{CTZ,DGIS}. In this paper we compare them with Gibbs measures of 2D Coulomb gases, as pioneered in \cite{Lau}, thereby relating this problem to the mean-field limit for classical electrostatics and to random matrices ensembles. 

It is convenient to work with scaled variables, defining (we do not emphasize the dependence on $m$)
\begin{equation}\label{eq:intro scaling}
\muN (Z) := N ^N \left| \PsiGV_m (\sqrt{N} Z )\right| ^2.
\end{equation}
The plasma analogy consists in comparing the one-body density corresponding to \eqref{eq:intro scaling}, i.e.
\[
\muNone (z) = \int_{\R ^{2(N-1)}} \muN (z,z_2,\ldots,z_N) dz_2\ldots dz_N, 
\]
with the minimizer $\rhoMF$ of the mean-field free energy functional\footnote{Here and in the sequel we will drop integration elements from integrals when there is no possible confusion.} 
\begin{equation}\label{eq:intro MFf}
\MFf [\rho] = \intR  \Vm \rho + 2 D(\rho,\rho) + N ^{-1} \int_{\R ^2} \rho \log \rho 
\end{equation}
amongst probability measures on $\R ^2$, $\rho\in \PP(\R ^2)$. Here 
\[
\Vm (r) = r ^2 - 2 \frac{m}{N} \log r  
\]
and the notation 
\begin{equation}\label{eq:2D Coulomb}
D(\rho,\rho) = -\iint_{\R^2\times \R^2} \rho(x) \log |x-y|  \rho(y) dxdy
\end{equation}
stands for the 2D Coulomb energy.

We will provide new estimates for this classical problem, discussed at length in Section \ref{sec:QHphases}. Theorem \ref{teo:MF limit} therein is our main result in this direction. Combined with Proposition \ref{pro:MF func} below it can be summarized as follows

\begin{teo}[\textbf{Plasma analogy for QH trial states}]\label{teo:plasma intro}\mbox{}\\
There exists a constant $C>0$ such that for large enough $N$ and any smooth function $V$ on $\mathbb R^2$
\begin{equation}\label{densitydiff}
\left\vert \intR V \left(\muNone - \rhoMF\right)  \right\vert \leq  CN^{-1/2} \log N \Vert \nabla V \Vert_{L ^2 (\R ^2)} + C N ^{-1/2}\Vert \nabla V \Vert_{L ^{\infty} (\R ^2)}
\end{equation}
if $m\lesssim N ^2$, and  
\begin{equation}\label{densitydiff2}
\left\vert \intR V \left(\muNone - \rhoMF \right)  \right\vert \leq  CN^{1/2} m ^{-1/4} \Vert V \Vert_{L ^{\infty} (\R ^2)}
\end{equation}
if $m\gg N ^2$.
\end{teo}

This result allows to evaluate the one-particle density of our trial states in a simple and explicit way. Of course it cannot be directly applied with $V= \pot(\sqrt{N} \: \cdot \: )$, i.e.,  the suitably rescaled physical potential, because of the growth at infinity of this potential. To compensate this growth we shall provide exponential decay estimates for the density of our trial states (see Theorem \ref{teo:QH phases}). These show that the contribution of large radii $r$ to \eqref{eq:intro pot ener} is negligible, so we can apply Theorem \ref{teo:plasma intro} to a suitable truncation of $V= \pot(\sqrt{N} \: \cdot \: )$ and deduce estimates of the potential energy \eqref{eq:intro pot ener}. Optimizing these over $m$, we find 
\begin{equation}\label{eq:intro m opt}
\mopt = \begin{cases}
               0 \mbox{ if } \om \geq - 2 k N \\
               - \frac{\om}{2 k} - N \mbox{ if } \om < - 2 k N.
              \end{cases} 
\end{equation}
We interpret this as a strong indication that, within the fully correlated regime, a transition occurs for $\om < 0$ and $|\om| \propto k N$ between a pure Laughlin state and a correlated state with a density depletion at the origin. Interestingly we also find that the character of the mean-field density $\rhoMF$ and thus that of the one-particle density of the state \eqref{eq:trial states} strongly depends on $m$: For $m\ll N ^2$ it is correctly approximated by a flat density profile located in a disc or an annulus (depending on the value of $m$), whereas for $m\gg N ^2$ the density profile is approximately a radial Gaussian centered on some circle. This is due to a transition from a dominantly electrostatic to a dominantly thermal behavior of the 2D Coulomb gas to which we compare our trial states. Using the expression of the optimal value of $m$ given in \eqref{eq:intro m opt} this suggests a further transition in the ground state of \eqref{eq:intro LLLe} in the regime $|\om|\propto k N ^2$ . Establishing 
these phenomena rigorously remains a challenging open problem. 

\medskip

The rest of the paper is organized as follows: In Section \ref{sec:model} we formulate \eqref{eq:intro LLLe} precisely as a minimization problem over the Bargmann space of analytic functions. In Section \ref{sec:results} we state our main results about this model, whose proofs are given in \cite{RSY}, and discuss in Section \ref{sec:trial intro} how Theorem \ref{teo:plasma intro} allows to interpret and improve them. Section \ref{sec:QHphases} is the mathematical core of the present paper. It contains the details on the plasma analogy, along with the proof of Theorem \ref{teo:plasma intro} and related results. Finally, Section~\ref{sec:energy} shows how to use the plasma analogy to obtain rigorous estimates of the potential energy of our QH trial states.

\section{Rotating bosons in the lowest Landau level}\label{sec:model results}

\subsection{The model}\label{sec:model}

In this section we discuss the derivation of the effective Hamiltonian that we will study and its most important properties. Our starting point, Equation \eqref{eq:intro ham 3D} below, is the full 3D Hamiltonian for a rotating trapped Bose gas with repulsive interactions given by a two-body  potential $v\geq 0$. We make three standard approximations: (1) the motion along the axis of rotation is frozen, (2) the states for the motion in the plane are reduced to the lowest Landau level, (3) the interaction potential is replaced by a delta function $g \delta$ where $g$ is proportional to the scattering length of $v$ (see \cite[Appendix C]{LSSY} for a definition).  

These approximations have been studied in \cite{LS} in the case $k=0$. Rigorous bounds quantifying their validity have been derived, and the approach can be adapted with almost no modifications to the case $k\neq 0$. For these reasons we only sketch the derivation of the reduced Hamiltonian and focus on its essential properties.

%\medskip

%\noindent\textbf{The full many-body Hamiltonian.} 

\subsubsection*{The full many-body Hamiltonian}

The rotating Bose gases, where it is proposed to try to create quantum Hall phases, can be described \emph{in the rotating frame} using the following many-body Hamiltonian
\begin{equation}\label{eq:intro ham 3D}
H^{\rm 3D}_N=\sum_{j=1}^N\left\{\half \left(\mathrm i\nabla_j+\mathbf A(\x_j)\right)^2 + V(\x_j) - \frac12 \Om ^2 (x_{j,1} ^2 + x_{j,2} ^2) \right\}
+  \sum_{ i < j } v(|\x_i - \x_j|).  
\end{equation}
Here  $\x_j=(x_{j,1},x_{j,2},x_{j,3})\in\mathbb R^3$ is the coordinate of the $j$-th particle, 
%${\mathbf L}=-\hbox{\rm i}\,\x\wedge{\mathbf \nabla}$ is the angular momentum operator,
$V$ is a confining external potential and $v$ the two-body interaction potential. Units are chosen so that $\hbar=m=1$. We also choose the coordinate axis so that the rotation vector $\mathbf\Omega=\Omega\mathbf e_3$ is proportional to the unit vector $\mathbf e_3$ in the 3-direction. The vector potential 
$$\mathbf A(\x)=\Omega (x_2,-x_1,0)$$ 
in \eqref{eq:intro ham 3D} represents the Coriolis force (analogous to the Lorentz force) while the negative quadratic potential proportional to $\Omega^2$ corresponds to to the centrifugal force. We are interested in the ground state of this Hamiltonian.

We consider a trapping potential of the form
\begin{equation}\label{eq:intro pot harm}
V(\x) = \frac12 \Omega_{\perp} ^2 \left(x_{1} ^2+x_{2} ^2 \right) + k \left(x_{1} ^2+x_{2} ^2 \right) ^2 + \frac{1}{2} \Omega_{\parallel} x_{3} ^2
\end{equation}
and it is necessary for $H^{\rm 3D}_N$ to be bounded below to require that the effective potential
\begin{equation}\label{eq:pot eff}
V_{\rm eff} (\x) =   V(\x) - \frac12 \Om ^2 (x_{1} ^2 + x_{2} ^2),
\end{equation}
taking into account the effect of the centrifugal force, remains bounded below. This is ensured if either $k>0$ or $k=0$ and $\Omega < \Omega_{\perp}$.

%\medskip

%\noindent\textbf{Reduction to the lowest Landau level.} 

\subsubsection*{Reduction to the lowest Landau level}

The 1-particle Hamiltonian in \eqref{eq:intro ham 3D} is given by
\begin{equation}\label{eq:onepartham} 
H_1=\half \left(\mathrm i\nabla  +\mathbf A(\x)\right)^2 +\half(-\partial_3^2+ \Omega_\parallel^2\,x_3^2) + V(\x) - \frac12 \Om ^2 (x_{1} ^2 + x_{2} ^2).
\end{equation} 
The first term is the 2D Landau Hamiltonian with spectrum $2(n+\half)\mathrm \Omega$, $n\in \N$, the second a one-dimensional harmonic oscillator in the 3-direction with spectrum $(n_\parallel+\half)\Omega_\parallel$, $n_\parallel\in \N$.  These two terms commute and can thus be diagonalized simultaneously.

When the energy scales associated with the effective trapping potential in the 12-plane
$$ \frac12 \Omega_{\perp} ^2 \left(x_{1} ^2+x_{2} ^2 \right) + k \left(x_{1} ^2+x_{2} ^2 \right) ^2- \frac12 \Om ^2 (x_{1} ^2 + x_{2} ^2) $$
are much smaller than the gaps between the energy levels of the first two operators in \eqref{eq:onepartham} it is natural to restrict attention to joint eigenstates of the first two terms of \eqref{eq:onepartham} with $n=n_\parallel=0$.
The motion in the 3-direction is  then `frozen' in the ground state of the harmonic oscillator\footnote{We could also consider a more general trapping potential in the 3-direction.}. As far as the 12-plane is concerned the state is in the lowest Landau level (LLL). 

Henceforth we choose units so that $\Omega=1$. Replacing $(x_1,x_2)$ by the complex variables $z=x_1+\mathrm i x_2$ and $\bar z=x_1-i x_2$  and denoting
$\partial_z=\half(\partial_1-\mathrm i\partial_2)$, $ \partial_{\bar z}  = \bar \partial _z =\half(\partial_1+\mathrm i\partial_2)$
we can write the Landau Hamiltonian as
\beq \half (\mathrm i\nabla_\perp+\mathbf A(\x))^2=2\left(a^\dagger a+\half\right)\eeq 
with the creation and annihilation operators
$a^\dagger=-\partial+\half\bar z$, $a=\bar \partial+ \half z$.
%These operators satisfy the canonical commutation relations $[a,a^\dagger]=1.$
Eigenfunctions $\psi(z,\bar z)$ in the lowest Landau level are solutions of  the equation $a\psi=0$, i.e.,
$$\bar \partial_z \psi(z,\bar z)=-\frac 12 z\psi(z,\bar z).$$ 
They are therefore of the form
\beq \label{bargmannfunctions}
\psi(z,\bar z)=f(z)\exp(-|z|^2/2)
\eeq 
with $\bar \partial _z f(z)=0$, i.e., $f$ is an analytic function of $z$.
   
%\medskip   
   
%\noindent\textbf{The Bargmann space}. 

\subsubsection*{The Bargmann space}

As seen above, single particle wave functions in the LLL correspond uniquely to functions in the Bargmann \cite{Bar,GJ} space $\mathcal B$ of  analytic functions $f$ on $\C$ such that
\beq 
\langle f,f\rangle_{\B}=\int_{\C} |{f(z)}|^2 \exp(-|z|^2)\, dz<\infty
\eeq 
where $dz$ denotes the Lebesgue measure on $\mathbb C\simeq \mathbb R^2$. We denote by $\LLL$ the space of the full wave functions $\psi$ including the Gaussian factor in \eqref{bargmannfunctions}. It is a subspace of the Hilbert space of square integrable functions w.r.t.  $dz$. Thus, state vectors in the LLL can either be regarded as elements of $\B$ or of $\LLL$ and we shall make use of both points of view.

For our $N$-body system of bosons in the LLL the corresponding Hilbert spaces are the symmetric tensor powers of $\B$ or $\LLL$, denoted by $\BargN$ and  $\LLLN$~: 
\begin{align}\label{eq:BargN}
\BargN &=  \left\{ F \mbox{ holomorphic and symmetric such that } F(z_1,\ldots,z_N) e^{-\sum_{j=1} ^N |z_j| ^2 / 2 } \in L ^2 (\R ^{2N})\right\}\\
\label{eq:LLLN}
\LLLN &= \left\{ \Psi (z_1,\ldots,z_N) = F(z_1,\ldots,z_N) e^{-\sum_{j=1} ^N |z_j| ^2 / 2 } \in L ^2 (\R ^{2N}),\: F \in \BargN \right\}. 
\end{align}
{Note that functions $\Psi$ in the LLL actually depend on both $z$ and $\bar{z}$ (because of the gaussian factor), but we believe that our notation does not lead to any confusion.
The scalar product on $\BargN$ is given by
\begin{equation}\label{eq:BargN scalar}
\bral F,\:G \ketr_{\BargN} = \bral F e^{-\sum_{j=1} ^N |z_j| ^2 / 2 },\:G e^{-\sum_{j=1} ^N |z_j| ^2 / 2 } \ketr_{L ^2 (\R ^{2N})}.
\end{equation}

%\medskip

%\noindent\textbf{The reduced $N$-body Hamiltonian.} 

\subsubsection*{The reduced $N$-body Hamiltonian}

We can now define our energy functional. When the above reductions have been made, the only term in the one-particle Hamiltonian \eqref{eq:onepartham} that is not fixed is the effective potential term in the $12$ plane. For a short range potential it makes sense physically to replace $v$ by a delta pair-potential $g \delta$. It also makes sense mathematically since the wave functions of the lowest Landau level are smooth. Justifying the substitution rigorously is a difficult task \cite{LS} on which we will not elaborate, but if we take this for granted we obtain
\begin{equation}\label{eq:ener LLL}
\LLLf [\Psi] = N \int_{\R ^2} \pot (r) \rhoP (z) dz + 4 N(N-1) g  \int_{\R^{2(N-1)}} \left| \Psi (z_2,z_2,z_3\ldots,z_N)\right| ^2 dz_2\ldots dz_N
\end{equation}
where $\rhoP$ is the one-body density normalized to $1$, $g$ the coupling constant and  
\begin{equation}\label{eq:pot}
\pot (r) = \om r ^2 + k r ^4. 
\end{equation}
The ground-state energy $\LLLe$ is then defined as in \eqref{eq:intro LLLe}. It is useful to reformulate this problem in the Bargmann space:

\begin{lem}[\textbf{Hamiltonian in the Bargmann space}]\label{lem:ham barg}\mbox{}\\
Define the Hamiltonian acting on $\BargN$
\begin{equation}\label{eq:Bargh}
\Bargh := N\left( \om + 2 k \right) + \sum_{j=1} ^N \left(\left( \om + 3 k \right) L_j + k L_j ^2 \right)+ g \sum_{1\leq i<j\leq N} \delta_{ij} 
\end{equation}
with $L _j = z_j \dd _{z_j}$ the angular momentum operator in the $j$-th variable and
\begin{equation}\label{eq:inter Barg} 
\delta _{ij} F (\dots,z_i,\dots,z_j\dots)=\frac 1{2\pi} F \big(\dots,\half(z_i+z_j), \dots,\half(z_i+z_j),\dots\big).
\end{equation} 
We have
\begin{equation}\label{eq:barg LLL plus}
\LLLf[\Psi] = \bral F,\Bargh F \ketr_{\BargN} \mbox{ with } \Psi = F \exp\left(-\sum_{j=1} ^N \frac{|z_j| ^2}{2} \right),\: F \in \BargN.   
\end{equation}
In particular
\begin{equation}\label{eq:barg LLL}
\LLLe = \Barge := \inf \sigma_{\BargN} \Bargh. 
\end{equation}
\end{lem}

\begin{proof}
The fact that the $\delta$ interaction operator acts as \eqref{eq:inter Barg} seems to have been noticed first in \cite{PB}. Indeed, taking \eqref{eq:inter Barg} as the definition of the interaction operator and using the analyticity of $F \in \BargN$, one easily realizes that
\begin{multline} \label{eq:link inter}
\langle F ,\delta _{ij} F \rangle_{\BargN} = \int_{\mathbb C^N}|F(\cdots,z,\cdots,z,\cdots)|^2\exp(-2|z|^2)\,dz \exp(-\hbox{$\sum_{k\neq i,j}$}|z_k|^2)\hbox{$\prod_{k\neq i,j}$}d z_k.
\end{multline}
It is sufficient to perform the computation for a two-body $f\in \Barg ^2$ :
\begin{align*}
\bral f, \delta_{12} f\ketr_{\Barg ^2} &= \frac{1}{2\pi} \iint_{\C ^2} \bar{f} (z_1,z_2) f\left( \frac{z_1+z_2}{2},\frac{z_1+z_2}{2}\right) e^{-|z_1| ^2 - |z_2| ^2}dz_1 dz_2 \\
&= \frac{2}{\pi} \int_{\C} f(u,u) e ^{-2|u| ^2} \int_{\C} \bar{f} (u+v,u-v) e ^{-2 |v| ^2} dv du.
\end{align*}
Fixing $u$, writing $g(v) = f(u+v,u-v)$ with an holomorphic $g(v) = \sum_{n\geq 0} g_n v^n$ we find 
$$ \int_{\C} \bar{f} (u+v,u-v) e ^{-2 |v| ^2} dv = \bar{g}_0 \int_{\C} e ^{-2 |v| ^2} dv = \frac{\pi}{2} \bar{g}(0) = \frac{\pi}{2} \bar{f}(u,u), $$
which is what we need.

On $\mathcal B$ the operator of the 3-component of the angular momentum is  $L_z=z\partial_z$ and an integration by parts shows (see e.g. \cite[Lemma 3.1]{AB}) that, for $f \in\mathcal B$,
\beq\label{eq:squarepot} 
\langle f,L_z f \rangle_{\B}=\int_{\C} ( |z|^2-1)|{f(z)}|^2 \exp(-|z|^2)\,dz. 
\eeq 
Integrating by parts twice we also obtain
\beq\label{eq:quartpot} 
\int_{\C}  |z|^4 |{f(z)}|^2 \exp(-|z|^2)\, dz = \langle f,L_z ^2 f \rangle_{\B} + 3 \langle f,L_z f \rangle_{\B} + 2 \langle f,f \rangle_{\B}.
\eeq 
Putting \eqref{eq:link inter}, \eqref{eq:squarepot} and \eqref{eq:quartpot} together proves \eqref{eq:barg LLL plus} and \eqref{eq:barg LLL}.
\end{proof}

Note that the angular momentum operator $z \dd_z$ on $\Barg$ has eigenvalues $\ell=0,1,2,\dots$ with corresponding normalized eigenfunctions $f_\ell(z)=(\pi \ell!)^{-1/2}\, z^\ell$. On $\LLL$ it acts as  $z\partial_z-\bar z\bar\partial_z$ rather than $z\partial_z$. 

%\medskip 

%\noindent \textbf{Existence of a ground state.} 

\subsubsection*{Existence of a ground state}

The essential virtue of writing the Hamiltonian as in Lemma \ref{lem:ham barg} is that the operator $\Bargh$ sends $\BargN$ onto $\BargN$. Clearly $(2\pi)\delta _{ij}$, defined in \eqref{eq:inter Barg}, is a projector, in particular a bounded operator. Also $\Bargh$ commutes with the total angular momentum operator
\begin{equation}\label{eq:Ltot}
\LL _N = \sum_{j=1} ^N z_j \dd_{z_j} 
\end{equation}
and we have
\begin{lem}[\textbf{Existence of a ground state}]\label{lem:existence}\mbox{}\\
$\Barge$, as defined in \eqref{eq:barg LLL} is an eigenvalue of $\Bargh$ with (possibly non unique) eigenfunction $\Bargm$. One may choose $\Bargm$ to have a definite angular momentum, say $L_0$:
$$ \LL_N \Bargm = L_0 \Bargm.$$
Equivalently, the infimum in \eqref{eq:intro LLLe} is attained. One may choose a minimizer $\LLLm$ with definite angular momentum $L_0$.
\end{lem}

\begin{proof}
Since $\Bargh$ commutes with $\LL_N$ we may look for its ground state by looking at the joint spectrum of $\Bargh$ and $\LL_N$,
\begin{equation}\label{eq:joint spectrum}
\Barge = \inf_{L\in \N} \: \inf \left\{\bral F , \Bargh F \ketr,\: F \in \BargN, \LL_N F = L F \right\}. 
\end{equation}
The subspace of $\BargN$ given by $\LL_N F = L F$ has finite dimension (it is spanned by the elementary symmetric polynomials of $N$ variables with total degree $L$) and thus the bottom of the spectrum of $\Bargh$ in this subspace is an eigenvalue. Obviously we have
$$ \Bargh \geq N\left( \om + 2 k \right) + \sum_{j=1} ^N \left(\left( \om + 3 k \right) L_j + k L_j ^2 \right)$$
and it is easy to deduce (see Lemma \ref{lem:low bound} below) that
$$\inf \left\{\bral F , \Bargh F \ketr,\: F \in \BargN, \LL_N F = L F \right\}  \to \infty \mbox{ when } L \to \infty.$$
The infimum over $L$ in \eqref{eq:joint spectrum} is thus attained, say at $L_0$, and the bottom of $\sigma_{\BargN} \Bargh$ is an eigenvalue, with eigenfunction (a priori not unique) $\Bargm$, which may be chosen with angular momentum $L_0$ . 
\end{proof}

%\medskip

%\noindent\textbf{The yrast curve and fully correlated states.} 

\subsubsection*{The yrast curve and fully correlated states}

In the case $k=0$, the Hamiltonian $\Bargh$ is simply 
\begin{equation}\label{eq:Bargh harmonic}
\Bargh = N\om + \om \LL_N + g \intN 
\end{equation}
with $\LL_N$ the total angular momentum \eqref{eq:Ltot} and $\intN$ the total interaction operator
\begin{equation}\label{eq:inter tot}
\intN = \sum_{1\leq i<j \leq N} \delta_{ij}. 
\end{equation}
This Hamiltonian has been studied extensively in the literature see e.g. the reviews \cite{Co,Vie} and \cite{LS,LSY,RSY} for references. Its essential property is that the operators $\mathcal L_N$ and{ $\mathcal I_N$ {commute}. The  lower boundary of (the convex hull of) their joint spectrum in a plot with angular momentum as the horizontal axis is called the  {\it yrast curve} (see \cite{RJ1,RJ2,RCJJ,VHR} and \cite{LS} for plots showing its qualitative features.). As a function of the eigenvalues $L$ of $\mathcal L_N$ the Yrast curve $I(L)$ is monotonously decreasing, starting at $I(0)= C N(N-1)$ and hitting zero at $L=N(N-1)$. The monotonicity follows from the observation that if a simultaneous eigenfunction of $\mathcal L_N$ and $\mathcal I_N$ is multiplied by the center of mass, $(z_1+\cdots+z_N)/N$, the interaction is unchanged while the angular momentum increases by one unit. 

For a given ratio $\omega/g$ the ground state of \eqref{eq:Bargh harmonic} (in general not unique) is determined by the point(s) on the yrast curve  where a supporting line has slope $-\omega/g$. For $L\leq N$ the ground state of \eqref{eq:Bargh harmonic} is explicitly known \cite{PB,SW} while for large $N$ and $L\ll N^2$ a Gross-Pitaevskii description with an uncorrelated ground state is asymptotically correct \cite{LSY}. For
$L=N(N-1)$ the unique ground state of $\mathcal I_N$ with eigenvalue $0$ is the bosonic {\it  Laughlin state} \eqref{eq:intro Laughlin} whose wave function in $\mathcal B_N$ is the symmetric polynomial
\beq \label{eq:Laughlin Barg}
F_{\rm Lau}(z_1,\dots,z_N)=c_{\rm Lau} \prod_{1\leq i<j\leq N}(z_i-z_j)^2
\eeq
with a normalization constant $c_{\rm Lau}$. More generally we have 

\begin{lem}[\textbf{Null space of the interaction operator}]\label{lem:null space}\mbox{}\\
The null space of the interaction operator \eqref{eq:inter Barg} is given by  
\begin{equation}\label{eq:kernel}
\Ker = \left\{  F_{\rm Lau} F,\: F \in \BargN \right\}.
\end{equation}
\end{lem}

\begin{proof}
Null states of $\intN$ must vanish on the diagonals $z_i = z_j$, hence contain the factor $\prod_{i<j}(z_i-z_j)$ by analyticity. Null state are thus of the form 
$$F(z_1,\dots, z_N) = \prod_{1\leq i<j\leq N}(z_i-z_j) \tilde{G}(z_1,\ldots,z_N)$$
and the bosonic symmetry of the wave function $F$ imposes that $\tilde{G}$ be anti-symmetric (fermionic) with respect to particles exchanges
$$\tilde{G} (\ldots,z_i,\ldots,z_j,\ldots) = - \tilde{G} (\ldots,z_j,\ldots,z_i,\ldots), \mbox{ for any } i,j $$
Then $\tilde{G}$ also has to vanish on the diagonals $z_i=z_j$, which leads to the form  
\begin{equation}\label{eq:kernel barg}
F(z_1,\dots, z_N) = \prod_{1\leq i<j\leq N}(z_i-z_j)^2 G(z_1,\ldots,z_N)
\end{equation}
by analyticity again and proves \eqref{eq:kernel}.
\end{proof}

We shall call states of the form \eqref{eq:kernel barg} {\it fully correlated states} because adding more correlations to them cannot decrease further the interaction energy. Their angular momentum spectrum is contained in $L\geq N(N-1)$ and amongst them only the Laughlin state has angular momentum exactly $N(N-1)$.

%\medskip

%\noindent\textbf{Spectral gaps of the interaction operator.} 

\subsubsection*{Spectral gaps of the interaction operator}

Clearly, a proof that a ground state of $\Bargh$ almost fully lives in $\Ker$ will rely on the operator $\intN$ having a spectral gap above its ground state. However this is not known at present although it is widely believed to be true. It is however possible to restrict our attention to states satisfying bounds on their angular momentum. For example, states having too large an angular momentum will be proved to also have an unreasonably large potential energy as compared to that of the trial states we are going to construct. Once restricted to states with finite angular momentum, the interaction operator is a non zero operator on a finite dimensional space and thus it trivially has a gap. Important quantities for us will be the gaps that one obtains when restricting the interaction operator in such a manner :
\begin{equation}\label{eq:gap}
\gap \left(L \right):= \min \left(\sigma \left( \intN | _{ \{ \LL_N = L \}} \right) \setminus \left\{ 0 \right\}\right)
\end{equation}
where $\LL_N$ is given by \eqref{eq:Ltot}. The function $L \mapsto \gap (L)$ is decreasing for the same reason that the yrast curve is.

Numerical diagonalizations of the interaction operator (see e.g. \cite{RCJJ,RJ1,RJ2,VHR}) for small numbers of particles suggest that the unrestricted operator truly has a gap, and that it is attained at angular momentum $N(N-1) - N$ : 
$$\gap \left( L \right) = \gap(N(N-1) -N)\mbox{ for any } L\geq N(N-1)-N.$$
It is also believed (see e.g. discussions in \cite{LS}) that this gap stays of order $1$ when $N\to +\infty$. Proofs of these conjectures are unknown to us and hence our criteria for the ground state of \eqref{eq:LLLh} being strongly correlated will depend on quantities whose $N$-dependence is not known. If true, the conjectures would imply a uniform lower bound
\[
\gap(L) \geq \gap (N(N-1)-N) \geq C \mbox{ for any } L\in \N  
\]
which would simplify the conditions on the parameters to have strong correlations in the ground state that we give below and make explicit their dependence on $N$.

\subsection{Main results}\label{sec:results}

With the notation above one can easily see \cite[Section 2.2]{LS} that when $k=0$ and $g/\om$ is large enough (in dependence with $N$), the ground state of $\Bargh$ is exactly given by the Laughlin state.
%\footnote{Note that $\gap(N(N-1) - 1) > 0 $ because all states with zero interaction are of the form \eqref{eq:kernel barg}}. 

When $k\neq 0$, the single particle part of $\Bargh$ is no longer a multiple of $\LL_N$ as in \eqref{eq:Bargh harmonic}. In particular it no longer commutes with the interaction operator and the Laughlin state is no longer an exact eigenstate of the full Hamiltonian. One may only hope that $\Bargm$ becomes almost fully correlated in an appropriate limit. Identifying such a limit is the main goal of the companion paper \cite{RSY}.
Note also that the spectrum of the single particle Hamiltonian in \eqref{eq:Bargh} is given (up to an additive constant) by 
$$\left( \omega + 3k \right)\ell + k\ell ^2, \quad \ell\in \N$$
with normalized eigenfunctions $f_\ell(z)=(\pi \ell!)^{-1/2}\, z^\ell$. When $\om < 0$, which is allowed if $k\neq 0$, the single particle part thus favors states with non zero angular momentum $\ell$, in contrast with the situation when $\om \geq 0$. One may thus imagine to obtain fully correlated states with angular momentum larger than $N(N-1)$. 

\medskip

We now state precisely our results, starting by providing estimates on the angular momentum of the ground state of \eqref{eq:LLLh}. In particular, this will provide the reduction to states having finite angular momenta we were alluding to before. Since our Hamiltonian commutes with the total angular momentum $\LL_N$ we may choose $\LLLm$ in a definite total angular momentum sector. We denote $L_0$ the total momentum of such a ground state :
\begin{equation}\label{eq:mom LLLm}
\sum_{j=1} ^N (z_j \dd_{z_j} - \bar{z}_j\bar{\dd}_{z_j}) \LLLm = L_0 \LLLm
\end{equation}
and we have
\begin{teo}[\textbf{Angular momentum estimates}]\label{teo:momentum}\mbox{}\\
In the limit $N\to \infty$, $\om,k \to 0$ the angular momentum $L_0$ of a ground state of $\LLLh$ satisfies
\begin{enumerate}
\item If $\om \geq - 2 k N$,
\begin{equation}\label{eq:result mom 1}
L_0 \leq 2 N ^2 
\end{equation}
\item If $\om \leq -2 kN$ and $|\om| /k \ll N ^2$,
\begin{equation}\label{eq:result mom 2}
\left| L_0 - \Lgv \right| \leq \sqrt{3} N ^2,
\end{equation}
where
\begin{equation}\label{eq:mom GV}
\Lgv =  - \frac{\om N}{2k} + O(1).
\end{equation}
In particular  $L_0/\Lgv \to 1$ if $ N \ll |\omega|/k\ll N 2$.
\item If $\om \leq -2 kN$ and $|\om| /k \gg N ^2$
\beq\label{eq:result mom 3}
\left| L_0 - \Lgv \right| \leq \sqrt{3} \Lgv ^{1/2} N. 
\eeq
In particular  $L_0/\Lgv \to 1$ if $|\omega|/k \gg N^{2}$.
\end{enumerate}
\end{teo}
                       
%Note that in the case $k=0$ it is obvious that $L_0 \leq N(N-1)$ : increasing further the angular momentum cannot decrease the interaction energy. 

The situation described in Item 1 of the above is compatible with the Laughlin state staying the ground state of the Hamiltonian. On the other hand, Items 2 and 3 show that the Laughlin state is \emph{not} the true ground state for sufficiently negative values of the ratio $\om / kN$: in this case the ground state has a momentum given to leading order by $\Lgv$ ($\rm qh$ stands for quasi-hole again). A state with much larger angular momentum than Laughlin's is thus favored. As we will discuss below, we prove that a wave function containing a vortex at the origin in addition to the correlations of the Laughlin state has a lower energy than the pure Laughlin state.   

\medskip

Our next theorem is a criterion for the ground state of \eqref{eq:Bargh} to be asymptotically fully correlated, in the sense that its projection $\PKerp (\LLLm)$ on the orthogonal complement of $\Ker$ (as defined in \eqref{eq:intro kernel}) vanishes in a certain limit. As anticipated above, our criteria depend on spectral gaps of the interaction operator. Recalling \eqref{eq:gap}, let us define 
\begin{align}\label{eq:defi gaps}
\Delta_{1} &=  \gap(2 N^2)\nonumber \\
\Delta_{3} &= \gap(\Lgv + \sqrt 3 N^{2})\nonumber \\
\Delta_4  &= \gap(\Lgv+ \sqrt 3 \Lgv ^{1/2} N),
\end{align}
where $\Lgv$ is defined as in \eqref{eq:mom GV}. The indices in the notation correspond to the different cases in the following theorem. A reader willing to take for granted the conjectures about the spectral gaps of $\intN$ we discussed above may replace these quantities by fixed numbers in the following statements.

\begin{teo}[\textbf{Criteria for strong correlations in the ground state}]\label{teo:correl}\mbox{}\\
Let $\LLLm$ be a minimizer of \eqref{eq:intro LLLe}. We have
\begin{equation}\label{eq:correl}
\left\Vert \PKerp \LLLm \right\Vert  \to 0  
\end{equation}
in the limit $N\to \infty$, $\om,k\to 0$ if one of the following conditions holds~:\smallskip

\noindent {\bf Case 1.} $\om \geq 0$ and 
\[
\left( g \,\Delta_1 \right)^{-1}(\omega N^2+kN^3)\rightarrow 0. 
\]
%In this case, for $N$ large enough, 
%\begin{equation}\label{eq:estim correl 1}
%\Vert \PKerp \LLLm\Vert^2\leq\frac 13 \left(a \,\Delta_{1,2} (C) \right)^{-1} kN^3 (1+o(1)).
%\end{equation}

\noindent {\bf Case 2.} $- 2 k N\leq \om \leq 0$ and
\[
\left(g \,\Delta_1  \right)^{-1}({N \om ^2}/{k} + \om N ^2 + k N ^3) \rightarrow 0. 
\]
%Again \eqref{eq:estim correl 1} holds under this condition for $N$ large.
\smallskip

\noindent{\bf Case 3.} $\om \leq - 2 k N$ and
\[
\left(g \: \Delta_3  \right)^{-1} k N ^{3}\rightarrow 0.
\]
%Here 
%\begin{equation}\label{eq:estim correl 2}
%\Vert \PKerp \LLLm \Vert^2\leq \frac 13 \left(a \: \Delta_3 (C) \right)^{-1} k\,N^3(1+o(1)).
%\end{equation}
\noindent{\bf Case 4.}
$\om \leq - 2 k N$ and
$$
\left(g \: \Delta_4 
\right)^{-1} |\omega|N\rightarrow 0.
$$
%Here 
%\beq\label{eq:estim correl 3}
%\Vert P^\perp \Psi_0\Vert^2\leq \frac 32 \left(a \: \Delta_4 (C) \right)^{-1} |\omega| N (1+o(1)).
%\eeq 
\end{teo}

Note that, given some $\om$ and $k$ depending on $N$ in a definite manner, one can always choose $g$ so large that one of the criteria in Theorem \ref{teo:correl} is satisfied in the limit $N \to \infty$.

\medskip

We now state some energy estimates. In the cases described in Theorem \ref{teo:correl} we are able to determine the order of magnitude of the leading order of the energy, though with unmatching constants.

\begin{teo}[\textbf{Energy bounds}]\label{teo:result ener}\mbox{}\\ 
The ground state energy $\LLLe$ satisfies the following bounds:
\medskip

\noindent {\bf Cases 1 and 2}
\begin{equation}\label{eq:result ener 1}
\left(\omega N^2+kN^3\right)(1-o(1))\leq \LLLe \leq\left( \omega N^2+\frac 43 kN^3\right)(1+o(1)).
\end{equation}
\noindent {\bf Case 3}
\begin{equation}\label{eq:result ener 2}
 -\frac {\omega^2N}{4k}(1-o(1))\leq \LLLe \leq \left(-\frac {\omega^2N}{4k}+\frac 13 kN^3\right)(1+o(1)). 
\end{equation}
\noindent {\bf Case 4}
\beq \label{eq:result ener 3}
-\frac {\omega^2N}{4k}(1-o(1))\leq \LLLe \leq \left(-\frac {\omega^2N}{4k}+\frac 32 |\omega| N\right)(1+o(1)).
\eeq
\end{teo}

The proofs of these results are given in \cite{RSY}. The rest of the paper is devoted to their interpretation, using the plasma analogy that we expose in details in Section \ref{sec:plasma}.

\subsection{Discussion}\label{sec:trial intro}

The reason for the occurrence of different cases in Theorems \ref{teo:correl} and \ref{teo:result ener} can be interpreted in view of the plasma analogy. Indeed, the minimization of \eqref{eq:intro MFf} is a simple electrostatics problem and accurate approximations to $\rhoMF$ may be computed, leading to the following picture:
\begin{itemize}
\item In case 1 the effective potential is increasing and it is favorable to use the pure Laughlin state as a trial state. We prove that in this case the mean-field density is approximately constant in a disc around the origin.
\item In cases 2 to 4, the effective potential has a local maximum at the origin and a minimum along some circle of radius $\ropt(\om,k)$. In case 2 the potential well along $r=\ropt (\om,k)$ is not deep enough to make it favorable to deplete the density of the trial state at the origin, and  the Laughlin state is still preferred.
\item For the cases 3 and 4, we notice that the mean-field density $\rhoMF$ is well approximated by a profile with a maximum along a circle of radius $r(m,N)$. Equating $\ropt(\om,k)$ and $r(m,N)$ in order that the maximum of the density coincides with the minimum of the potential we find an optimal choice for the phase circulation $m$ of the giant vortex at the origin 
\begin{equation}\label{eq:m opt intro}
\mopt = \begin{cases}
               0 \mbox{ if } \om \geq - 2 k N \\
               - \frac{\om}{2 k}  - N \mbox{ if } \om < - 2 k N.
              \end{cases} 
\end{equation}
The Laughlin state is thus favored for $\om \geq -2kN$ whereas there is a tendency towards adding a vortex at the origin in the opposite regime.
\item The character of the density of the optimal trial state changes from an almost constant profile when $m\ll N ^2$ to a Gaussian profile when $m\gg N ^2$, corresponding to the change in the order of magnitude of the subleading contribution to the energy upper bound in Theorem \ref{teo:result ener} that distinguishes case 3 from case 4. 
\end{itemize}

The fact that the mean-field density profile changes in the regime $m\propto N ^2$ is the reason why the estimates of Theorem \ref{teo:plasma intro} are stated differently when $m\ll N ^2$ and $m\gg N ^2$. The change in the physics is reflected by the need of a different approach to the mean-field limit in the two regimes.

We thus see that the plasma analogy provides a rationale for the occurrence of different cases in the minimization of the energy functional describing rotating bosons in the lowest Landau level when the trapping potential is of the form \eqref{eq:intro pot eff} (more complicated expressions could be considered). As we prove below (Section \ref{sec:energy}), it also allows to improve some of our energy upper bounds. The rest of the paper is devoted to the proof of Theorem \ref{teo:plasma intro} and to the study of the mean-field energy functional \eqref{eq:intro MFf}.

\section{Quantum Hall states and the plasma analogy}\label{sec:QHphases}

In this core section of the paper we use the interpretation of the modulus squared of fully correlated trial states as the Gibbs measure of a 2D Coulomb gas (one-component plasma) to compute the single-particle density of the trial states in the limit of large particle number. This interpretation has been instrumental since the first introduction of the Laughlin state \cite{Lau,Lau2} in the context of the fractional quantum Hall effect, see \cite{Gir} for a review. Ideas derived from the plasma analogy are commonly used in the literature (see e.g. \cite{BCR,CTZ,DGIS,GRG,LFS,Jan,JLS}), and we shall provide a rigorous justification to some of them.

It is convenient to scale distances by a factor $\sqrt{N}$: 
\begin{equation}\label{eq:Gibbs states}
\muN (Z) := N ^N \left| \PsiGV_m (\sqrt{N} Z )\right| ^2.
\end{equation}
With such a rescaling, we can recognize the Gibbs measure of a 2D Coulomb gas with temperature $T= N ^{-1}$ and a mean-field scaling in the interactions ($\ZN$ is a normalization factor)
\begin{eqnarray}\label{eq:plasma analogy}
\muN (Z) &=& \ZN ^{-1} \exp\left( \sum_{j=1} ^N \left( - N  |z_j| ^2 + 2 m \log |z_j|\right)  - 4 \sum_{i<j} \log |z_i - z_j|\right)\nonumber \\
&=& \ZN ^{-1} \exp\left( -\frac{1}{T} \left( \sum_{j=1} ^N \left( |z_j| ^2 - 2 \frac{m}{N} \log |z_j|\right) - \frac{4}{N} \sum_{i<j} \log |z_i - z_j|\right)\right)\nonumber \\
&=& \ZN ^{-1} \exp\left( -\frac{1}{T} \HN \right),
\end{eqnarray}
where the Coulomb Hamiltonian $\HN$ is defined as
\begin{equation}\label{eq:Coul hami}
\HN (Z) : =  \sum_{j=1} ^N \Vm (z_j) - \frac{2}{N} \sum_{i\neq j} \log |z_i - z_j|
\end{equation}
with 
\begin{equation}\label{eq:Vm}
\Vm (z) = |z| ^2 - 2 \frac{m}{N} \log |z|.
\end{equation}
This model describes $N$ classical 2D particles located at points $z_1,\ldots,z_N$ in the complex plane, interacting via 2D Coulomb forces and feeling the electric potential generated by a constant background of opposite charge (the $|z_j| ^2$ terms). When $m\neq 0$, the term $-2\frac{m}{N} \log |z_j|$ describes the effect of a particle of charge $2\frac{m}{N}$ fixed at the origin. Our classical one-component plasma is thus more precisely a jellium  with an additional point charge pinned at the origin.

Note the $1/N$ factor in front of the interaction term: the interest of scaling the distances is to put us in a mean-field regime. Common wisdom about the thermodynamic limit for classical particles then suggests that we shall be able to extract information about $\muN$ from a limit $N\to \infty$. More precisely, one should expect that $\muN$ factorizes
\begin{equation}\label{eq:factor formal}
\muN \approx \rho ^{\otimes N} \mbox{ when } N\to \infty, 
\end{equation}
for some well-chosen probability measure $\rho \in \PP (\R ^2)$ (see Section \ref{sec:plasma} below), in the sense that
\begin{equation}\label{eq:factor formal prec}
\muN ^{(k)} \approx \rho ^{\otimes k} \mbox{ when } N\to \infty \mbox{ and } k \mbox{ is fixed}.
\end{equation}
Here we denote by $\mu ^{(k)}$ the $k$ particle density of a symmetric measure $\mu$ over a Cartesian product, defined by integrating $\mu^{(k)}$ over $N-k$ variables:
\begin{equation}\label{eq:defi marginal}
\mu  ^{(k)} (z_1,\ldots,z_k) := \int_{\R ^{2(N-k)}} \mu (z_1,\ldots,z_N) dz_{k+1}\ldots dz_{N}.
\end{equation}
Results in this direction are given in \cite{CLMP,MS,Kie1,KS} for related models. A large deviation result is presented in \cite{BZ}. Adapting these methods we could prove that $\muN ^{(k)} \wto \rho ^{\otimes k}$ weakly as measures for any fixed $k$. For our purpose, however,  quantitative estimates are needed and we thus use a different method.

\medskip

Our new approach to the mean-field limit of the Coulomb gas works in any scaling of the spatial variables but it is important to note that once a scaling has been chosen, the dependence of the temperature on $N$ is fixed. In the most convenient mean-field scaling that we have chosen above, the temperature is 
\begin{equation}\label{eq:set T}
T=N ^{-1} 
\end{equation}
and thus $T\to 0$ when $N\to \infty$. One could thus expect the plasma, to which we compare the density of our QH states, to be close to being in its ground state. These heuristic considerations have to be taken with caution however since the potential $\Vm$ depends on both $m$ and $N$. As it turns out there is a transition in the physics of our trial state, the dividing line being given by $m\propto N ^2$: For $m\ll N ^2$, electrostatic effects dominate, i.e. the plasma is close to its ground state at zero temperature, whereas for $m\gg N ^2$ entropy considerations dominate the physics, i.e. the temperature plays an important role. 

Although we shall not use this analogy, it is worth recalling the strong connection between Coulomb gases and Gaussian random matrices, noted first by Wigner, see \cite{For} for references. Interestingly, in the random matrix context, one also has to analyze a Coulomb gas with mean-field interactions and temperature of order $N^{-1}$, see \cite{KS}. The study of the 2D Coulomb gas problem, or some of its generalizations, is also related to vortex systems in classical and quantum fluids, see e.g. \cite{CLMP,CY,SS} for discussions. 

\subsection{The mean-field one component plasma} \label{sec:plasma}

Let us now go into more details about our approach to the mean-field limit. As is well-known \cite{CLMP,Kie1} and easy to prove, $\muN$ minimizes the free-energy functional (recall the temperature is $N ^{-1}$)
\begin{equation}\label{eq:free ener f}
\FN [\mu] := \intRN \HN (Z) \mu(Z) dZ + T \intRN \mu(Z) \log \mu(Z)dZ
\end{equation}
amongst symmetric probability measures $\mu \in \PP_s (\R^{2N})$. We denote by
\begin{equation}\label{eq:free ener e}
\FNe = \FN[\muN] = - T \log \ZN
\end{equation}
the minimum free energy.

We will prove estimates relating the minimization of \eqref{eq:free ener f} to that of the mean-field free energy functional
\begin{equation}\label{eq:MFf}
\MFf [\rho] = \intR  \Vm \rho + 2 D(\rho,\rho) + T \int_{\R ^2} \rho \log \rho 
\end{equation}
amongst probability measures $\rho\in \PP(\R ^2)$. We denote by $\rhoMF$ and $\MFe$ respectively the ground state and the ground state (free) energy of the mean-field free energy functional. The notation 
\begin{equation}\label{eq:2D Coulomb 2}
D(\rho,\rho) = -\iint_{\R^2\times \R^2} \rho(x) \log |x-y|  \rho(y) dxdy
\end{equation}
stands for the 2D Coulomb energy.

As usual, $\MFf$ is obtained by restricting $\FN$ to trial states of the form $\rho^{\otimes N}$, which should be a reasonable approximation when $N$ is large. The main goal of this section is to justify this approximation by proving that
\[
\FNe \approx N \MFe, \quad \muN \approx \rhoMF \: ^{\otimes N}
\]
in a sense to be made precise below, and with quantitative estimates.

The transition in the physics of our trial states between dominantly electrostatic and thermal behaviors can be taken into account by introducing two simplified functionals that will be used to approximate \eqref{eq:MFf} in the two different regimes. In the electrostatic regime we drop the entropy term and define
\begin{equation}\label{eq:MFfh}
\MFfel [\rho] = \intR  \Vm \rho + 2 D(\rho,\rho)
\end{equation}
with ground state $\rhoMFel$ and ground state energy $\MFeel$ whereas in the thermal regime we drop the electrostatic term to obtain
\begin{equation}\label{eq:MFfth}
\MFfth [\rho] = \intR \Vm \rho + T \intR \rho \log \rho 
\end{equation}
whith ground state $\rhoMFth$ and ground state energy $\MFeth$. For our computational purpose it is much more convenient to estimate the difference between $\muNone$ and $\rhoMFel$ or $\rhoMFth$ because, at least with the relatively simple potentials $\Vm$ we consider, the two latter functions are explicit, see Proposition~\ref{pro:MF func} below.

%As it turns out, the difference between $\rhoMFel$ and $\rhoMF$ is roughly of the same order as the error we will make in approximating $\muNone$ by $\rhoMF$ in the mean-field limit procedure. 

The main output of the plasma analysis is the following theorem. We denote 
\[
 \ropt = \sqrt\frac{m}{N}
\]
the minimum point of the potential $\Vm$.

\begin{teo}[\textbf{Plasma analogy for quantum Hall phases}]\label{teo:QH phases}\mbox{}\\
There exists a constant $C>0$ such that we have 
\begin{enumerate}
\item (Mean-field limit in the electrostatic regime). For $m\lesssim N ^2$ and any $V$ such that $\nabla V\in L ^{\infty} (\R ^2)\cap L ^{2} (\R ^2)$  
\begin{equation}\label{eq:1 particle lim el}
\left\vert \intR \left(\muNone - \rhoMFel \right) V \right\vert \leq  C N ^{-1/2} (\log N ) ^{1/2} \Vert \nabla V \Vert_{L ^2 (\R ^2)} + C N ^{-1/2}\Vert \nabla V \Vert_{L ^{\infty} (\R ^2)}.
\end{equation}
Also, for some constants $c,C>0$ and for $N$ large enough
\begin{equation}\label{eq:1 particle decay}
\muNone (z) \leq C \exp\left( -c N \left(\left( |z|-\ropt \right)^2 -  \log N \right)\right) \mbox{ when } ||z|-\ropt|\geq C \max (N ^{1/2} m ^{-1/2}, N ^{-1/2} ). 
\end{equation}
\item (Mean-field limit in the thermal regime). For  $m\gg N ^2$ and any $V\in L ^{\infty} (\R ^2)$
\begin{equation}\label{eq:1 particle lim th}
\left\vert \intR \left(\muNone - \rhoMFth \right) V \right\vert \leq  C N ^{1/2} m ^{-1/4} \Vert V \Vert_{L ^{\infty} (\R ^2)}.
\end{equation}
Moreover there exists $c,C>0$ such that
\begin{equation}\label{eq:1 particle decay thermal}
\muNone(z) \leq \exp\left(-c N (|z|-\ropt) ^2\right)\mbox{ when } ||z|-\ropt| \geq C N^{1/2} m ^{-1/4}.
\end{equation}

%\item (Decay) for some constant $c>0$ and for $N$ large enough
%\begin{equation}\label{eq:1 particle decay}
%\muNone \leq \rhoMF e^{cN \log N} 
%\end{equation}
\end{enumerate}
\end{teo}

The proof of this theorem relies on upper and lower bounds to the free energy proving that $\FNe \sim N \MFe$ with controlled error. Our approach can also give information on the reduced densities $\muN ^{(k)}$, for $k$ fixed in the limit $N\to \infty$ see Remark \ref{rem:marginals} below. We state only \eqref{eq:1 particle lim el} and \eqref{eq:1 particle lim th} explicitly because they will be our main tools for estimating the energy of our quantum Hall trial states. What these equations say is that we can replace $\muNone$ by $\rhoMFel$ or $\rhoMFth$, making a controlled error. Note that to put them to good use in the proofs of our main results we will have to truncate the physical potential $\pot$ so that the norms appearing in the right-hand sides be finite. To estimate the error this induces we need to know that $\muNone$ has a suitable decay, which is the purpose that \eqref{eq:1 particle decay} and \eqref{eq:1 particle decay thermal} serve. As we will prove below, $\rhoMF$ decays rather fast in the 
regions where $\rhoMFel$ and $\rhoMFth$ are small so that one can hope not to make a large error when truncating the physical potential.  

The proof of this result goes as follows :  In Section \ref{sec:MF func} we first study the mean-field functional and prove that $\rhoMF$ can be approximated by $\rhoMFel$ (respectively $\rhoMFth$) when $m\ll N ^2$ (respectively when $m\gg N ^2$). We also study the decay of $\rhoMF$, which will provide the desired decay of $\muNone$ in the electrostatic regime according to \eqref{eq:1 particle decay}. In Section \ref{sec:QH thermo} we study the mean-field limit and thereby relate $\muNone$ to $\rhoMF$, which will complete the proof of \eqref{eq:1 particle lim el}. Most of the arguments in this part apply to much more general situations than that we are directly interested in. The interested reader should have no difficulty in adapting our proofs to different potentials than our specific $\Vm$, and to other temperature regimes than $T= N^{-1}$. Finally Section \ref{sec:MF thermal} contains the proof of \eqref{eq:1 particle lim th}. As explained below we have to follow different strategies for the 
electrostatic and thermal regimes, which accounts for the different norms of $V$ appearing in the right-hand sides of \eqref{eq:1 particle lim el} and \eqref{eq:1 particle lim th}, and the different forms of the decay estimates \eqref{eq:1 particle decay} and \eqref{eq:1 particle decay thermal}.

\subsection{The mean-field functionals}\label{sec:MF func}

We now state several facts about the mean-field problems, some being well-known from potential theory (see \cite{ST} for references). 
%Most important to us will be the estimate \eqref{eq:dif rhoMF} on the difference between $\rhoMF$ and $\rhoMFel$, which proof uses very simple but apparently previously unnoticed stability estimates, equations \eqref{eq:stability MF} and \eqref{eq:stability MFh} below. 

\begin{pro}[\textbf{The mean-field functionals}]\label{pro:MF func}\mbox{}\\
The following properties hold
\begin{enumerate}
\item (Existence, Uniqueness). The functionals \eqref{eq:MFf}, \eqref{eq:MFfh} and \eqref{eq:MFfth} each admit a unique minimizer among probability measures, respectively denoted $\rhoMF$, $\rhoMFel$ and $\rhoMFth$. Moreover 
\begin{equation}\label{eq:bounds rhoMF}
0 <\rhoMF \leq \frac{1}{2\pi} \mbox{ a.e. in } \R ^2. 
\end{equation}
\item (Electrostatic regime). We have the explicit expression
\begin{eqnarray}\label{eq:exp rhoMFmh 0}
\rhoMFel &=& \frac{1}{2\pi} \one_{B(0,\sqrt{2})} \mbox{ if } m=0 \\ 
\rhoMFel &=& \frac{1}{2\pi} \one_{\AN} \mbox{ if } m>0 \label{eq:exp rhoMFmh}
\end{eqnarray}
where $\AN$ is the annulus of inner radius $\Rminus = \sqrt{m/N}$ and outer radius $\Rplus = \sqrt{2 + m/N}$ centered at the origin. Moreover
\begin{equation}\label{eq:MF elec regime}
D \left(\rhoMFel- \rhoMF,\rhoMFel- \rhoMF \right) \leq C N ^{-1}.
\end{equation}
\item (Thermal regime). We have the explicit expression
\begin{equation}\label{eq:exp rhoMFmth}
\rhoMFth (r) = \frac{1}{\Zth} \exp \left( - \frac{1}{T} \Vm (r) \right)  = \frac{2N^{m+1}}{\pi m!} |z|^{2m} e ^{-N|z|^2}
\end{equation}
where $\Zth$ is a normalization constant satisfying $\MFeth = - T \log \Zth.$  Moreover,  
\begin{equation}\label{eq:MF ther regime coul}
D \left(\rhoMFth- \rhoMF,\rhoMFth- \rhoMF \right) \leq C m ^{-1/2}
\end{equation}
for any $m\gg N ^2$ and 
\begin{equation}\label{eq:MF ther regime TV}
 \left\Vert \rhoMFth- \rhoMF \right\Vert_{\rm TV} \leq C \frac{N ^{1/2}}{m ^{1/4}} 
\end{equation}
where $\left\Vert \mu \right\Vert_{\rm TV} = \int |\mu_+| + \int|\mu_-|$ stands for the total variation norm of a measure $\mu$.
\end{enumerate}
\end{pro}

\begin{rem}\label{rem:com MF}
\begin{enumerate}
\item \emph{Norms}. Note that $D(.,.)$ is the square of a norm on the space of measures with total mass $0$ as we will see below. It is actually the square of the $\dot{H}^{-1}$ norm ($L^2$ norm in Fourier space with weight $|k|^{-2}$). Comparing \eqref{eq:MF elec regime} and \eqref{eq:MF ther regime coul} one can see that $\rhoMF$ is better approximated by $\rhoMFel$ (respectively $\rhoMFth$) when $m\ll N^2$ (respectively $m \gg N ^2$). When $m\gg N ^2$, it becomes possible  to use the total variation norm to estimate the difference between $\rhoMF$ and $\rhoMFth$ \eqref{eq:MF ther regime TV}, which is more convenient for practical purposes. Note that in our case $\rhoMF$ and $\rhoMFth$ are $L ^1$ functions, so their total variation norm coincides with their $L^1$ norm.
\item \emph{Comparing the profiles}. When $m \gg N$ it is safe to approximate $\Vm$ with its second variation around $\ropt  = m^{1/2} N ^{-1/2}$, noting that $\Vm'' (\ropt) = O(1)$, which means that $\rhoMFth$ is roughly speaking a Gaussian profile centered on the minimum of $\Vm$. Obviously this is a very different shape from the electrostatic profile, which is constant in an annulus close to $\ropt$. More important are the scales involved in the two profiles: the electrostatic density has its maximum of order $1$ and consequently its support has a thickness of order $m ^{-1/2} N ^{1/2}$, whereas the thermal profile has a maximum of order $m ^{-1/2} N$ and thus is spread over an annulus of thickness $N ^{-1/2}$ to ensure normalization. 
\item \emph{Heuristics for the electrostatic/thermal transition}. A good criterion for the transition, that can be backed with energetic considerations, is the comparison of the length scales: To favor the potential energy, $\rhoMF$ wants to be as concentrated as possible close to the minimum of $\Vm$, the meaning of ``possible'' being set by the other terms in the functional. The entropy and electrostatic terms are associated with different length scales, $N^{-1/2}$ for the entropy and $m ^{-1/2} N ^{1/2}$ for the Coulomb term. In order to minimize the energy, the true profile $\rhoMF$ is spread over the maximum of these two length scales, i.e. on the electrostatic length scale for $m\ll N  ^2$ and on the thermal length scale for $m\gg N ^2$.
\end{enumerate}
\end{rem}
\hfill\qed

\medskip

Let us begin by recalling some well-known lemmas that we shall use several times in the sequel :

\begin{lem}[\textbf{Positivity of relative entropies and CKP inequality}]\label{lem:rel ent}\mbox{}\\
Let $\mu$ and $\nu$ be two probability measures with $\mu$ absolutely continuous with respect to $\nu$. Then
\begin{equation}\label{eq:pos rel ent}
\int \mu \log \frac{\mu}{\nu} \geq 0. 
\end{equation} 
More precisely one has the Csisz\'{a}r-Kullback-Pinsker (CKP) inequality 
\begin{equation}\label{eq:CKP}
\int \mu \log \frac{\mu}{\nu} \geq \frac{1}{2}   \left\Vert \mu - \nu \right\Vert_{\rm TV} ^2.
\end{equation}
\end{lem}

\begin{proof}
A simple application of Jensen's inequality : 
\[ 
\int \mu \log \frac{\mu}{\nu} = \int \nu \frac{\mu}{\nu} \log \frac{\mu}{\nu} \geq \left(\int \nu \frac{\mu}{\nu} \right) \log \left( \int \nu \frac{\mu}{\nu} \right) = 0,
\]
since $\mu$ and $\nu$ are probability measures and $x\mapsto x \log x$ is convex. A proof of the CKP inequality and some generalizations may be found in \cite{BV}.
\end{proof}

\begin{lem}[\textbf{Positivity properties of the 2D Coulomb energy}]\label{lem:Coul def pos}\mbox{}\\
Let $\mu$ be a Radon measure over $\R^2$ whose positive and negative parts $\mu_+$ and $\mu_-$  satisfy $|D(\mu_+,\mu_+)| < \infty$, $|D(\mu_-,\mu_-)| < \infty$. If 
\[
\int_{\R ^2} \mu = 0 
\]
then 
\begin{equation}\label{eq:positivity Coulomb}
D(\mu, \mu) \geq 0 
\end{equation}
with equality if and only if $\mu = 0$.

Consequently, the functional $\mu \mapsto D(\mu,\mu)$ is strictly convex on the convex set $\PP(\R^2)$ of probability measures on $\R ^2$.

\end{lem}

\begin{proof}
This is a consequence of the formula 
\[
 D(\mu,\mu) = \frac{1}{2\pi} \int_{\R ^2} \left( \int_{\R ^2} \frac{1}{|t-z|} d\mu(z) \right) ^2 dt
\]
that holds whenever $\intR \mu = 0$, see \cite[Chapter I, Lemma 1.8]{ST}.

To see that \eqref{eq:positivity Coulomb} implies the claimed convexity property, pick $\mu_1,\mu_2 \in \PP (\R ^2)$ and notice that
\[
\frac{1}{2} D(\mu_1,\mu_1) +\frac{1}{2} D(\mu_2,\mu_2) - D\left(\frac{1}{2}\mu_1+ \frac{1}{2}\mu_2, \frac{1}{2}\mu_1+ \frac{1}{2}\mu_2\right) = \frac{1}{4} D\left( \mu_1- \mu_2, \mu_1 - \mu_2\right) \geq 0
\]
since $\int_{\R ^2} \mu_1 = \int_{\R ^2} \mu_2 = 1$.
\end{proof}

\begin{lem}[\textbf{Newton's theorem}]\label{lem:Newton}\mbox{}\\
For a measure $\mu$ let  
\begin{equation}\label{eq:pot gen}
h_{\mu} (x)= - \intR \log |x-y| \mu(dy) 
\end{equation}
be the potential generated by $\mu$. If $\mu$ is radial then
\begin{equation}\label{eq:newton}
h_{\mu} (x) = -\log |x| \int_{|y|<|x|} \mu(dy) - \int_{|y|>|x|} \log |y| \mu (dy).
\end{equation}
\end{lem}

\begin{proof}
Simply reproduce the proof of the corresponding result in 3D, see \cite[Theorem 9.7]{LL}. 
\end{proof}

\noindent \emph{Proof of Theorem \ref{pro:MF func}.}

\emph{Step 1. (Existence, Uniqueness).}
The existence part is standard material, as are the following Euler-Lagrange equations  
\begin{align}
4 \hMF + \Vm +  N ^{-1} \log (\rhoMF) &= \MFe + 2 D(\rhoMF,\rhoMF) \mbox{ on } \R ^2 \label{eq:EEL MFm}\\
\Vm + N^{-1} \log (\rhoMFth) &= \MFeth \mbox{ on } \R ^2 \label{eq:EEL MFmth}. 
\end{align}
Strictly speaking, studying the variations of $\MFf$ and $\MFfth$ only tells that the above equations hold on the support of $\rhoMF$, respectively $\rhoMFth$. However, as we prove below, $\rhoMF$ and $\rhoMFth$ are strictly positive a.e. which is why the above equations hold on $\R ^2$.

On the contrary, $\rhoMFel$ has compact support (see below), which makes the Euler-Lagrange equation a little bit more subtle \cite{ST}:
\begin{align}
4 \hMFel + \Vm  &= \MFeel + 2 D(\rhoMFel,\rhoMFel) \mbox{ on } \supp (\rhoMFel) \label{eq:EEL MFmh 1}\\
4 \hMFel + \Vm  &\geq \MFeel + 2 D(\rhoMFel,\rhoMFel) \mbox{ on } \R ^2 \setminus \supp (\rhoMFel) \label{eq:EEL MFmh 2}.
\end{align}
Here $\hMF$ and $\hMFel$ are defined as in \eqref{eq:pot gen}.
%\begin{equation}\label{eq:hMFm}
%\hMF = 2 \pi \left( - \Delta \right) ^{-1} \rhoMF
%end{equation} 
%and 
%\begin{equation}\label{eq:hMFmh}
%\hMFel = 2 \pi \left( - \Delta \right) ^{-1} \rhoMFel
%\end{equation}
%the potentials generated by $\rhoMF$ and $\rhoMFel$ respectively. 
The value of the constants on the right-hand sides of the Euler-Lagrange equations is evaluated by multiplying the equations by $\rhoMF$, $\rhoMFth$ and $\rhoMFel$ respectively, and integrating. The formula \eqref{eq:exp rhoMFmth} is a direct consequence of \eqref{eq:EEL MFmth} once one knows that $\rhoMFth >0$ a.e. and thus that \eqref{eq:EEL MFmth} holds on the whole space.

The fact that we have equality on the whole of $\R ^2$ in \eqref{eq:EEL MFm} and \eqref{eq:EEL MFmth} follows from the fact that $\supp(\rhoMF) = \supp(\rhoMFth)= \R ^2$, which is probably a point that deserves a little discussion. We follow an argument from \cite{Ner} (proof of Proposition 15 therein) : Suppose for contradiction that $\supp(\rhoMF) ^c$ contains a set $S$ of nonzero Lebesgue measure. Consider the trial state 
\[
 \rho = \frac{\rhoMF + \ep \one_{S}}{1+\ep |S|}
\]
for some $\ep$ small enough. Evaluating $\MFfel[\rho]$ is easy and we find that for small enough $\ep$ there is a constant $C>0$ such that (we consider $N$, $m$ and $T$ as fixed here)
\[
\MFfel[\rho] \leq \MFfel[\rhoMF] + C \ep.
\]
To compute the entropy of $\rho$, the key point to notice is that since $\rhoMF$ and $\ep \one_{S}$ have disjoint supports we have
\begin{multline*}
\intR  \left(\rhoMF + \ep \one_{S}\right) \log \left( \rhoMF + \ep \one_{S} \right) = \intR \rhoMF \log  \rhoMF + \ep \intR \one_{S} \log \left( \ep \one_{S} \right) \\ 
\leq  \intR \rhoMF \log  \rhoMF + C \ep \left(1 +\log \ep \right),
\end{multline*}
and thus
\[
\MFf[\rho] \leq \MFf[\rhoMF] + C \ep \left(1 +\log \ep \right) < \MFe
\]
for $\ep$ small enough, which is a contradiction. The proof that also $\supp (\rhoMFth) = \R ^2$ is identical.

To see that $\rhoMF \leq (2\pi) ^{-1}$ a.e. in $\R^2$, we take the Laplacian of \eqref{eq:EEL MFm} and obtain  
\[
- T\frac{\Delta \rhoMF}{\rhoMF} + T \frac{|\nabla \rhoMF| ^2}{(\rhoMF) \:^2} + 8\pi \rhoMF - 4 + \frac{4\pi m}{N} \delta_0 = 0 
\]
which implies 
\[
- T \Delta \rhoMF + 8\pi (\rhoMF ) ^2 - 4 \rhoMF \leq 0. 
\]
At any local maximum of $\rhoMF$ we have $\Delta \rhoMF \leq 0$ and thus
\[
\rhoMF \leq \frac{1}{2\pi}, 
\]
which proves the claim.

\medskip 

The uniqueness of the minimizer of $\MFf$ follows from the strict convexity of the functional. Quantitatively, we have the following stability identity: For any probability measure $\rho$ that we write as $ \rho = \rhoMF + \nu$,
\begin{align}\label{eq:stability MF}
\MFf[\rho] &= \intR \Vm \rhoMF +\intR \Vm \nu + 2 D (\rhoMF, \rhoMF ) + 4 D (\rhoMF, \nu )+ 2 D(\nu,\nu)  \nonumber \\
&+  T \intR \left(\rhoMF + \nu \right) \log \left( \rhoMF + \nu \right) \nonumber \\ 
&= \intR \Vm \rhoMF + 2 D (\rhoMF, \rhoMF ) + T \intR \left(\rhoMF \right) \log \left( \rhoMF \right)  + 2 D(\nu,\nu) \nonumber \\ 
&+  T \intR \left( \left(\rhoMF + \nu \right) \log \left( \rhoMF + \nu \right) - \nu \log \rhoMF -  \rhoMF\log \rhoMF \right) \nonumber \\
&= \MFe + 2 D(\nu,\nu) +  T \intR  \left(\rhoMF + \nu \right) \log \left( \frac{\rhoMF + \nu}{\rhoMF} \right) 
\end{align}
where we have used \eqref{eq:EEL MFm} and the fact that $\intR \nu = 0$ to go to the second line. This yields uniqueness for $\rhoMF$ using Lemmas \ref{lem:rel ent} and \ref{lem:Coul def pos}. Note for later use that this also proves stability of the minimizer in $\dot{H}^{-1}$ and $\rm TV$ norms.

\medskip

\emph{Step 2. (The electrostatic profile).} The proof of \eqref{eq:exp rhoMFmh 0} and \eqref{eq:exp rhoMFmh} relies on \eqref{eq:EEL MFmh 1}. We start with the easiest case $m=0$ where we obtain the circle law for the Ginibre ensemble (see \cite{KS} for example). A proof of \eqref{eq:exp rhoMFmh 0} may be found in \cite{ST}, we give details for the convenience of the reader. Taking the Laplacian of \eqref{eq:EEL MFmh 1} we have
\begin{equation}\label{eq:rhoMFmh supp}
 \rhoMFel = \frac{1}{2\pi} \mbox{ on } \supp (\rhoMFel). 
\end{equation}
By uniqueness of the minimizer of $\MFfel$ and the radiality of $\Vm$,  $\rhoMFel$ must be radial. Its support is thus in any case a union of some annuli (counting a disc as an annulus with inner radius zero). Suppose for contradiction that the support is not a disc; Then there is a nonempty annulus $\A$ (or a disc centered at the origin as a special case) in the complement of the support and enclosed by the support. The potential
\begin{equation}\label{eq:varphi}
 \varphi = 4 \hMFel + \Vm - \MFeel - 2 D(\rhoMFel,\rhoMFel)
\end{equation}
satisfies $\varphi =0$ on $\dd \A$ by \eqref{eq:EEL MFmh 1} and $-\Delta \varphi = - 2 < 0$ in $\A$ since $\rhoMFel= 0$ there by definition. We deduce by the maximum principle that $\varphi \leq 0$ and thus $\varphi = 0$ in $\A$  because by \eqref{eq:EEL MFmh 2} we already know that $\varphi \geq 0$. Taking the Laplacian of the equation $\varphi = 0$  on $\A$ we would conclude as before that $\rhoMFel = (2\pi)^{-1}$ in $\A$ which is a contradiction with the fact that $\A \cap \supp (\rhoMFel) = \varnothing$. We conclude that the support of $\rhoMFel$ must be a disc centered at the origin, which implies \eqref{eq:exp rhoMFmh 0} via \eqref{eq:rhoMFmh supp} and the normalization constraint.

The result in the case $m>0$ follows from the same kind of arguments. We have from \eqref{eq:EEL MFmh 1}
\[
 \rhoMFel = \frac{1}{2\pi} - \frac{m}{2 N} \delta_0 \mbox{ on } \supp (\rhoMFel) 
\]
and since $\rhoMFel\geq 0$, we deduce that $0 \notin \supp (\rhoMFel)$. By the same maximum principle argument as above we deduce that $\supp (\rhoMFel)$ is an annulus and that \eqref{eq:rhoMFmh supp} also holds in the case $m>0$. We consider again $\varphi$ as defined in \eqref{eq:varphi}. By \eqref{eq:EEL MFmh 1}, it must be that $\varphi \equiv 0$ on $\supp (\rhoMFel)$. In particular it is constant there. On the other hand, 
\[
\varphi = - \MFeel - 2 D(\rhoMFel,\rhoMFel) + \tilde{\varphi} = \mathrm{const} + \tilde{\varphi}
\]
is a constant plus the potential $\tilde{\varphi} = 4 \hMFel + \Vm$ generated by 
\begin{itemize}
\item a constant background of charge density $- (2 \pi) ^{-1}$ (coming from the $|z| ^2$ term in $\Vm$)
\item a point charge of strength $\frac{2m}{N}$ located at the origin (coming from the $2\frac{m}{N} \log|z|$ term in $\Vm$)
\item the charge density $ 4 \rhoMFel$, equal to $(2 \pi) ^{-1}$ in an annulus of radii say $\Rminus$ and $\Rplus$ and $0$ elsewhere.
\end{itemize}
Using Newton's theorem, Lemma \ref{lem:Newton}, this potential, evaluated at any $ \Rminus \leq r\leq \Rplus$ equals that generated by a point charge of strength $\frac{2m}{N} - 2\Rminus \:^2$ located at the origin plus another constant (the constant background and the charge density $ 4 \rhoMFel$ cancel each other in $\supp (\rhoMFel)$). The only possibility for $\tilde{\varphi}$, and therefore $\varphi$, to be constant in $\supp (\rhoMFel)$ is then to have 
\[
\Rminus =  \sqrt{\frac{m}{N}}
\]
so that the effective potential generated on $\supp(\rhoMFel)$ by the constant background in $B(0,\Rminus)$ is canceled (screened) by the point charge sitting at the origin. Using the normalization of $\rhoMFel$ and \eqref{eq:rhoMFmh supp}, we compute $\Rplus$ and \eqref{eq:exp rhoMFmh} follows.

\medskip

\emph{Step 3. (Electrostatic regime).} We now turn to the proof of \eqref{eq:MF elec regime}. First, taking $\rhoMFel$ as a trial state for $\MFf$ we have
\[
\MFe \leq \MFeel + T \intR \rhoMFel \log \rhoMFel
\]
which, in view of \eqref{eq:exp rhoMFmh 0} and \eqref{eq:exp rhoMFmh}, yields
\begin{equation}\label{eq:sup MFme}
\MFe \leq \MFeel -\log (2\pi)  T.
\end{equation}
Next, we denote
\[
\rho_0 (z)= \pi ^{-1} \exp(- |z| ^2)  
\]
and note that $\intR \rho_0 = 1$. Using Lemma \ref{lem:rel ent} we have 
\begin{eqnarray}
\MFe &=& 2 D(\rhoMF,\rhoMF) + \intR \Vm \rhoMF + T \intR \rhoMF \log \rho_0 + T  \intR \rhoMF \log \frac{\rhoMF}{\rho_0} \nonumber\\
%&\geq & 2 D(\rhoMF,\rhoMF) + \intR \left( \Vm - T|z| ^2 \right) \rhoMF + T \log c_0 \nonumber\\
&\geq & 2 D(\rhoMF,\rhoMF) + \intR \left( \Vm - T |z| ^2 \right) \rhoMF - \log (\pi) T.
\label{eq:MF low bound elec}
\end{eqnarray}
Now, the functional 
\[
\MFftel[\rho] := 2 D(\rho,\rho) + \intR \left( \Vm - T |z| ^2 \right)\rho
\]
is of the same type as $\MFfel$, at least if $N$ is large enough for the $- T |z| ^2$ term in the above to be smaller than the $|z| ^2$ term in $\Vm$ (recall that we have $T= N ^{-1}$). We denote by $\MFetel$ and $\rhoMFtel$ respectively the ground-state energy and the minimizer of $\MFftel$. Using the Euler-Lagrange equation satisfied by $\rhoMFtel$, similar to \eqref{eq:EEL MFmh 1}, and a computation analogous to \eqref{eq:stability MF} we obtain for any measure $\rho$
\begin{equation}\label{eq:stability elec}
\MFftel[\rho] \geq \MFetel + 2 D \left(\rho- \rhoMFtel,\rho- \rhoMFtel \right).
\end{equation}
Note that in this case there is no entropy term in the analogue of \eqref{eq:stability MF} and the second line becomes an inequality because $\rhoMFtel$ has compact support and thus it satisfies an Euler-Lagrange equation only on its support. Outside of the support we have an inequality as in \eqref{eq:EEL MFmh 2}. We deduce from \eqref{eq:MF low bound elec} and \eqref{eq:stability elec} that
\begin{equation}\label{eq:inf MFme}
\MFe \geq \MFetel + 2 D \left(\rhoMFtel- \rhoMF,\rhoMFtel- \rhoMF \right) -C T.
\end{equation}
Using the explicit expression for $\rhoMFtel$, similar to \eqref{eq:exp rhoMFmh}, it is then not difficult to see that also 
$$\MFetel = \MFeel + O(T),$$
and thus \eqref{eq:sup MFme} and \eqref{eq:inf MFme} combine to give 
\[
D \left(\rhoMF- \rhoMFtel,\rhoMF- \rhoMFtel \right) \leq C T.
\] 
We deduce that \eqref{eq:MF elec regime} holds by noting that also
\[
 D \left(\rhoMFtel- \rhoMFel,\rhoMFtel- \rhoMFel \right) \leq C T
\]
which can be proved easily, inspecting the explicit expressions for $\rhoMFtel$ and $\rhoMFel$.

\medskip

\emph{Step 4. (Thermal regime).} Taking $\rhoMFth$ as a trial state for $\MFfth$ we have
\begin{equation}\label{eq:MF up bound therm}
\MFe \leq \MFeth + 2 D (\rhoMFth,\rhoMFth). 
\end{equation}
To obtain a lower bound we write
\[
D(\rhoMF,\rhoMF) = D (\rhoMF-\rhoMFth,\rhoMF-\rhoMFth) -  D(\rhoMFth,\rhoMFth) +2 D (\rhoMFth,\rhoMF)
\]
and deduce 
\begin{equation}\label{eq:MF low bound therm}
\MFe \geq \MFftth [\rhoMF] + 2 D (\rhoMF-\rhoMFth,\rhoMF-\rhoMFth) -  2 D(\rhoMFth,\rhoMFth) 
\end{equation}
where (we denote $h_{\rhoMFth}$ the potential associated to $\rhoMFth$) 
\[
\MFftth[\rho]:= \intR \left( \Vm + 4 h_{\rhoMFth} \right) \rho + T \intR \rho \log \rho 
\]
with ground state $\rhoMFtth$ and ground state energy $\MFetth$. Of course 
\begin{equation}\label{eq:thermal profile tilde}
\rhoMFtth = \frac{1}{\Ztht} \exp\left(-T ^{-1} (\Vm + 4 h_{\rhoMFth})\right) 
\end{equation}
for some normalization constant $\Ztht$ satisfying  
\begin{equation}\label{eq:Etht}
\MFetth = - T \log \Ztht.
\end{equation}
Now, $h_{\rhoMFth}$ is radial and satisfies on $\R ^2$
\[
 -\Delta h_{\rhoMFth} = \rhoMFth \geq 0.
\]
Integrating this equation over $B(0,r)$ and using Stokes' theorem we deduce 
\[
2\pi r \dd_r h_{\rhoMFth}(r) = - \int_{B(0,r)} \rhoMFth 
\]
and thus $|\nabla h_{\rhoMFth}(t) | \leq C r^{-1}$ for any $t\in B\left(r/2,2r\right)$. From this we deduce the estimate
\begin{equation}\label{eq:vari hthermal t}
\left|h_{\rhoMFth} (r) - h_{\rhoMFth} (\ropt)\right| \leq C \ropt ^{-1} |r-\ropt| \mbox{ for any } r\in B\left(\frac12 \ropt,2\ropt\right).  
\end{equation}
Since on the other hand $\Vm-\min \Vm$ grows as $C(r-\ropt) ^2$ close to $\ropt$, it is easy to deduce from \eqref{eq:thermal profile tilde} that $\rhoMFtth$ is exponentially small in the region where $\rhoMFth$ is, that is for $|r-\ropt| \gg N ^{-1/2}$. Indeed, note that in this region
\[
|\Vm-\min \Vm| \propto C|r-\ropt| ^2 \gg (N/m) ^{1/2} |r-\ropt| \propto \ropt ^{-1} |r-\ropt| \propto \left|h_{\rhoMFth} (r) - h_{\rhoMFth} (\ropt)\right|
\]
provided $m\gg N ^2$. Simple estimates then show that 
\begin{align}\label{eq:aaa1}
\MFetth &= - T \log \Ztht = - T \log \left[ \intR \exp\left(-T ^{-1} \Vm\right) \exp\left( - T ^{-1} \left(4 h_{\rhoMFth} (\ropt) + O(m ^{-1/2})\right)\right) \right]\nonumber\\
&= - T \log \left( \intR \exp\left(-T ^{-1} \Vm\right) \right) + 4 h_{\rhoMFth} (\ropt) + O(m ^{-{1/2}})\nonumber \\
&= \MFeth + 4 h_{\rhoMFth} (\ropt) + O(m ^{-{1/2}}).
\end{align}
On the other hand, similar considerations based on \eqref{eq:exp rhoMFmth} and \eqref{eq:vari hthermal t} lead to
\begin{equation}\label{eq:aaa2}
D(\rhoMFth,\rhoMFth) = \intR h_{\rhoMFth} \rhoMFth = h_{\rhoMFth} (\ropt) + O(m ^{-1/2})
\end{equation}
where the last term is $\ropt ^{-1} = \sqrt{N/m}$ times the length scale of $\rhoMFth$ (we are basically saying that $\rhoMFth$ resembles a delta function concentrated along the circle of radius $\ropt$). Coming back to \eqref{eq:MF low bound therm}, using \eqref{eq:aaa1} and \eqref{eq:aaa2}, we have thus proved that for $m\gg N ^2$
\[
 \MFe \geq \MFeth + 2 D(\rhoMFth,\rhoMFth) + 2 D (\rhoMF-\rhoMFth,\rhoMF-\rhoMFth) + O(m ^{-{1/2}}),
\]
which we combine with \eqref{eq:MF up bound therm} to obtain \eqref{eq:MF ther regime coul}.

To prove \eqref{eq:MF ther regime TV} we go back to \eqref{eq:MF low bound therm} again and note that for any $\rho \in \PP (\R ^2)$ we have 
\[
 \MFftth [\rho] \geq \MFetth + T \intR \rho \log \frac{\rho}{\rhoMFtth} \geq \MFetth +  \frac{T}{2} \left\Vert \rho - \rhoMFtth \right\Vert_{\rm TV}^2 
\]
as a consequence of the explicit expression of $\rhoMFtth$ and the CKP inequality \eqref{eq:CKP}. Combining this with the considerations above, our lower bound can be improved to 
\[
\MFe \geq \MFeth + 2 D(\rhoMFth,\rhoMFth) + 2 D (\rhoMF-\rhoMFth,\rhoMF-\rhoMFth) + \frac{T}{2} \left\Vert \rhoMF - \rhoMFtth \right\Vert_{\rm TV}^2 + O(m ^{-{1/2}}). 
\]
Combining with the upper bound \eqref{eq:MF up bound therm} we deduce 
\[
\left\Vert \rhoMF - \rhoMFtth \right\Vert_{\rm TV} \leq C T ^{-1/2} m ^{-1/4} = C N ^{1/2} m ^{-1/4}
\]
and \eqref{eq:MF ther regime TV} follows by estimating the difference between $\rhoMFtth$ and $\rhoMFth$, using \eqref{eq:vari hthermal t} and the explicit expressions \eqref{eq:exp rhoMFmth} and \eqref{eq:thermal profile tilde}.\\
\hfill \qed

\medskip

As announced, the proof of \eqref{eq:1 particle decay} requires an estimate of the decay of $\rhoMF$. This is the content of the following

\begin{pro}[\textbf{Decay of the mean-field density}]\label{pro:decay MF}\mbox{}\\
There exists a $C>0$ such that for any $r\in \R$ satisfying $|r-\ropt| > C \max (N ^{1/2} m ^{-1/2}, N ^{-1/2})$
\begin{equation}\label{eq:decay MF}
\rhoMF (r) \leq C  \exp(-C N (r-\ropt) ^2 ).
\end{equation}
\end{pro}

\begin{proof}
\emph{Step 1.} We start by proving that the potential
\begin{equation}\label{eq:defi delta pot}
h_{\rhoMFel-\rhoMF} = \hMFel -\hMF  = 2\pi \left( - \Delta \right) ^{-1} \left( \rhoMFel - \rhoMF \right)
\end{equation}
is in $\dot{H} ^1 (\R ^2)$ and that one has the bound 
\begin{equation}\label{eq:H1 dot bound}
\intR \left|\nabla h_{\rhoMFel-\rhoMF} \right| ^2 \leq C T. 
\end{equation}
Let us denote $\hat{f}$ the Fourier transform of a function $f$. We have (formally for the moment)
\[
\hat{h}_{\rhoMFel-\rhoMF} (k) = \frac{2\pi}{|k| ^2} \left( \rhoMFF (k) - \rhoMFelF (k)\right). 
\]
Now, since both $\rhoMF$ and $\rhoMFel$ are uniformly bounded in $L^1$ and in $L ^{\infty}$, they also are uniformly bounded in $L^2$, which implies that $\rhoMFF - \rhoMFelF$ is uniformly bounded in $L ^2$. On the other hand, since $\intR \Vm \rhoMF$ and $\intR \Vm \rhoMFel$ are finite, we deduce 
\[
\intR |x| \rhoMF (x)dx < \infty,\: \intR |x| \rhoMFel (x) dx< \infty
\]
which implies that $\nabla \left( \rhoMFF - \rhoMFelF\right) \in L ^ {\infty} (\R ^2)$. It remains to recall that 
\[
\rhoMFF (0) - \rhoMFelF (0) = \intR  \left( \rhoMF - \rhoMFel \right) = 0
\]
to deduce that, for $|k|$ small enough
\[
\left| \rhoMFF (k) - \rhoMFelF (k) \right| \leq C |k|. 
\]
Together with the fact that $\rhoMFF - \rhoMFelF \in L^2 (\R^2)$ this implies that 
\[
\intR \left|\nabla h_{\rhoMFel-\rhoMF} \right| ^2 = 4 \pi^2 \intR \frac{1}{|k| ^2} \left| \rhoMFF (k) - \rhoMFelF (k)\right| ^2 < \infty. 
\]
We can thus justify the integration by parts leading to  
\[
\intR \left|\nabla h_{\rhoMFel-\rhoMF} \right| ^2 = \intR -\Delta h_{\rhoMFel-\rhoMF} \left( \rhoMFel-\rhoMF \right) = D \left( \rhoMFel-\rhoMF, \rhoMFel-\rhoMF\right)
\]
and \eqref{eq:H1 dot bound} then follows from \eqref{eq:MF elec regime}.

\medskip

\emph{Step 2.} We now claim that for any $r\in \R$
\begin{multline}\label{eq:decay potential}
4 h_{\rhoMFel-\rhoMF} (r) \leq \MFeel + 2 D(\rhoMFel,\rhoMFel) - \MFe - 2 D(\rhoMF, \rhoMF) \\ + C \left( T + T ^{1/2} \frac{(r-\ropt) ^{1/2} (r+\ropt) ^{1/2}}{\min (\ropt,r)}\right).
\end{multline}
First note that taking the difference of equations \eqref{eq:EEL MFm} and \eqref{eq:EEL MFmh 1} we obtain 
\begin{align*}
4 h_{\rhoMFel-\rhoMF} (\ropt) &=  \MFeel + D(\rhoMFel,\rhoMFel) - \MFe - D(\rhoMF, \rhoMF) + T \log \rhoMF (\ropt) \\
&\leq \MFeel + D(\rhoMFel,\rhoMFel) - \MFe - D(\rhoMF, \rhoMF) + C T
\end{align*}
by using \eqref{eq:bounds rhoMF}. Using radiality we then have
\begin{eqnarray*}
h_{\rhoMFel-\rhoMF} (r)  &=& h_{\rhoMFel-\rhoMF} (\ropt) + \int_{\ropt} ^r h_{\rhoMFel-\rhoMF}' (t) dt   
\\ &\leq& h_{\rhoMFel-\rhoMF} (\ropt) + \frac{1}{\min(\ropt,r)} \int_{\ropt} ^r |\nabla h_{\rhoMFel-\rhoMF} | tdt 
\\ &\leq& h_{\rhoMFel-\rhoMF} (\ropt) + \frac{C T ^{1/2}}{\min(\ropt,r)} \left( \frac{r ^2}{2} - \frac{\ropt ^2}{2} \right) ^{1/2}
\end{eqnarray*}
where we use \eqref{eq:H1 dot bound} and the Cauchy-Schwarz inequality. This proves \eqref{eq:decay potential}. 

\medskip

\emph{Step 3.} Using \eqref{eq:EEL MFm} we have 
\[
\rhoMF = \exp \left( \frac{1}{T} \left( \MFe + 2 D(\rhoMF, \rhoMF) - \Vm - 4 \hMF \right)\right). 
\]
Inserting \eqref{eq:decay potential} and using $h_{\rhoMFel-\rhoMF}  = h_{\rhoMFel}-h_{\rhoMF} $ we deduce
\begin{multline}\label{eq:decay MF 2}
\rhoMF (r) \leq \exp\left[ \frac{1}{T} \left( \MFeel + 2 D(\rhoMFel,\rhoMFel) - \Vm(r) - 4 \hMFel (r) \right)\right] \\
\times \exp\left[ \frac{C}{T} \left(T + T ^{1/2}\frac{(\ropt + r ) ^{1/2}(\ropt-r) ^{1/2}}{\min (\ropt,r)}\right)\right].
\end{multline}
Then the exponential fall-off in \eqref{eq:decay MF} is provided by the decay of 
\begin{equation}\label{eq:discuss decay}
\MFeel + 2 D(\rhoMFel,\rhoMFel) - \Vm(r) - 4 \hMFel (r)
\end{equation}
away from the support of $\rhoMFel$ as a consequence of \eqref{eq:EEL MFmh 2}. More precisely, since \eqref{eq:EEL MFmh 1} implies 
\[
\MFeel + 2 D(\rhoMFel,\rhoMFel) = \Vm(\Rplus) + 4 \hMFel (\Rplus), 
\]
we see that \eqref{eq:discuss decay} decays as 
\[
\Vm(\Rplus) + 4 \hMFel (\Rplus)- 4 \hMFel (r) - \Vm (r) 
\]
for $r\geq \Rplus$ (we only detail this case, the proof is the same for the region $r\leq \Rminus$). Reasoning as when proving \eqref{eq:vari hthermal t}, one easily sees that  
\[
 \left|4 \hMFel (\Rplus)- 4 \hMFel (r)\right| \leq C \ropt ^{-1} |r-\Rplus|\leq C m ^{-1/2} N ^{1/2} |r-\Rplus|.
\]
On the other hand, approximating $\Vm$ by its second variation around $\ropt$ we have
\[
\Vm(\Rplus) - \Vm (r) \approx C (r-\ropt) ^2 -C (\Rplus-\ropt) ^2 = - C \left( (r-\Rplus) ^2 + 2 (r-\Rplus) (\Rplus-\ropt) \right)
\]
and since $|\Rplus-\ropt|=O(m ^{-1/2} N ^{1/2})$ we deduce that
\begin{equation}\label{eq:discuss decay 2}
\MFeel + 2 D(\rhoMFel,\rhoMFel) - \Vm(r) - 4 \hMFel (r)\leq - C (r-\Rplus) ^2
\end{equation}
for $r\geq \Rplus + c \:m ^{-1/2} N ^{1/2}$ with well-chosen $c$ and $C$.

This decay compensates for the other terms in \eqref{eq:decay MF 2} as soon as                                                                                                                                                                                                                                                                              
\begin{equation}\label{eq:length scales}
|r-\Rplus|  \gg \max \left( \sqrt{\frac{N}{m}}, \sqrt{T}, (mN) ^{-1/6} \right).
\end{equation}
Note that the first length scale in the max above is of the order of magnitude of the thickness of the support of $\rhoMFel$. The second length scale accounts for the mass spreading due to the entropy term and the third is associated with the ``error'' term 
\[
T ^{1/2} \frac{(\ropt + r ) ^{1/2}(\ropt-r) ^{1/2}}{\min (\ropt,r)} = \frac{(\ropt + r ) ^{1/2}(\ropt-r) ^{1/2}}{N ^{1/2} \min (\ropt,r)}. 
\]
Indeed, for example when $r$ is sufficiently close to $\ropt$,
\[
\frac{(\ropt + r ) ^{1/2}(\ropt-r) ^{1/2}}{N ^{1/2} \min (\ropt,r)} \propto (N\ropt) ^{-1/2} |r-\Rplus| ^{1/2}  \propto (mN) ^{-1/4} |r-\Rplus| ^{1/2},
\]
using the approximations $r \approx \ropt$ and $|r-\ropt|\approx |r-\Rplus|$ and recalling that $\ropt= \sqrt{m/N}$. To see that the third term in the right-hand side of \eqref{eq:length scales} is really an error, recall that we use \eqref{eq:decay MF} when $T= N ^{-1}$. There are then two cases 
\begin{itemize}
\item $m \ll N ^2$, which implies 
\[
\sqrt{\frac{N}{m}} \gg (mN) ^{-1/6} \gg \sqrt{T}, 
\]
i.e. the electrostatic length scale dominates the error and the entropic length 
\item $m \gg N ^2$, in which case
\[
\sqrt{T} \gg (mN) ^{-1/6} \gg \sqrt{\frac{N}{m}}, 
\]
i.e. the entropic length dominates the error and the electrostatic length.
\end{itemize}
In both cases the error term is dominated by either the electrostatic or the entropic term. Recalling that $|\ropt-\Rplus|\leq C N ^{1/2} m ^{-1/2}$ we thus have proved that \eqref{eq:decay MF} holds if
\[
m \ll N ^2 \mbox{ and } |r-\ropt| \gg \sqrt{\frac{N}{m}}   
\]
or 
\[
m \gg N ^2 \mbox{ and } |r-\ropt| \gg N ^{-1/2}, 
\]
which is the desired result.
\end{proof}

\subsection{Thermodynamic limit in the electrostatic regime}\label{sec:QH thermo}

We now turn to the study of the large $N$ limit of \eqref{eq:free ener f}. This is a rather classical question, especially since we are in a mean-field scaling. A line of attack for this kind of statistical mechanics problems has been pioneered in \cite{MS} for regular interparticle interactions and then carried on independently in \cite{Kie1,CLMP} in the case of logarithmic interactions. These works deal with the regime $T= O(1)$ in which all three terms in \eqref{eq:MFf} are of the same order of magnitude. They also consider the case of negative temperature which is more involved but irrelevant in our context. In connection with several ensembles of random matrices, \cite{KS} extends this approach to the regime $T \propto N ^{-1}$ and more general Hamiltonians. %This method has been recently adapted in order to treat the Hartree limit for bosonic atoms \cite{Kie2}.

Common to these approaches is a compactness argument, which does not lead to quantitative estimates on the precision of the mean-field approximation. As far as we know, it is only very recently \cite{SS} that constructive estimates have been obtained, in the specific case of the 2D Coulomb gas that we consider here. We could employ some of these estimates (in particular Theorem 3 therein) in our context, but they would not be sufficient. We thus prefer to implement a new method that gives different estimates, more suited to our purpose, with a simpler proof.

\begin{teo}[\textbf{Mean field limit for 2D Coulomb gases}]\label{teo:MF limit}\mbox{}\\
There exists a constant $C>0$ such that, for $N$ large enough, we have
\begin{enumerate}
\item Upper bound. 
\begin{equation}\label{eq:up bound}
\FNe \leq N \MFe - D (\rhoMF,\rhoMF)
\end{equation}
\item Lower bound.
\begin{equation}\label{eq:low bound}
\FNe \geq N \MFe - \frac{\log N}{2} - C  
\end{equation}
\item Estimate on the first marginal of the Gibbs measure.\\
For any $V\in \dot{H} ^1 (\R ^2)$ with $\nabla V \in L ^{\infty} (\R ^2)$ we have
\begin{equation}\label{eq:1 particle lim}
\left\vert \intR \left(\muNone - \rhoMF \right) V \right\vert \leq  C \left(\frac{\log N}{N}\right) ^{1/2} \Vert \nabla V \Vert_{L ^2 (\R ^2)} + C N ^{-1/2}\Vert \nabla V \Vert_{L ^{\infty} (\R ^2)}.
\end{equation}  
\end{enumerate}
\end{teo}

\begin{rem}\mbox{}
\begin{enumerate}
\item This result is not limited to the particular type of potential $\Vm$ we consider in the plasma analogy. Our method can accommodate any potential as soon as the associated mean-field functional and its minimizers are reasonably well-behaved. Note also that the estimates apply to any $T$, the two most interesting regimes being $T= O(1)$ (more natural from the Coulomb gas point of view) and $T = O(N ^{-1})$ (in relation with random matrices and quantum Hall phases). 
\item In \cite{SS} a different scaling convention is used, their $\beta$ being given as $\beta = N^{-1} T^{-1}$ in our units. The approach therein is limited to $\beta \geq O(1)$ and a fixed potential $\Vm$. It shows that the $- \log (N) / 2$ term gives exactly the second order correction in the regime $\beta \geq O(1)$ (i.e. $T \propto N^{-1}$) with fixed potential. It moreover gives the exact third order correction to the free energy $\FNe$ in the limit $\beta \to \infty$, that is $T \ll N^{-1}$ (see \cite[Theorem 1]{SS}). This is more difficult, since this connection is related to a Coulombian renormalized energy, whose definition is rather complex. The approach we develop here allows to recover the lower bound on $\FNe$ in the regime where $\beta$ is bounded above, i.e. $T\propto N^{-1}$. To see this, compare Theorem \ref{teo:MF limit} with Theorem 1 in \cite{SS}, keeping in mind that when $T\propto N^{-1}$, the entropy term in $\MFe$ is a lower order correction : $\MFe = \MFeel + O(N ^{-1})$.
\item In the regime $T\propto N ^{-1}$, we could most likely adapt arguments from \cite[Section 4 and 7]{SS} in order to construct a trial state capturing exactly the $- \log (N) / 2$ in the upper bound. This would be rather technical because, in contrast with what is assumed in \cite{SS}, our mean-field densities do depend on $N$ when $m\neq 0$. We thus content ourselves with a non optimal upper bound that has only a marginal impact on our main theorems. %The optimal upper bound would change the order of magnitude of the first term of the right-hand side of \eqref{eq:1 particle lim} to $N^{-1/2}$.
\end{enumerate}
\hfill \qed
\end{rem}

We need two classical lemmas. The first is the 2D version of the so-called Onsager lemma (see e.g. \cite[Lemma 6.1]{LieSei}):

\begin{lem}[\textbf{2D Onsager lemma}]\label{lem:Onsager}\mbox{}\\
Let $\mu$ be a radial probability measure on $\R ^2$. Denote, for some $l>0$ 
$$\muxi (z) = \mu\left( \frac{z-x_i}{l}\right).$$
We have, for any $\rho$ such that $\intR \rho = N$ and any $(x_1,\ldots,x_N)\in \R ^{2N}$
\begin{equation}\label{eq:Onsager}
- \sum_{i\neq j} \log |x_i-x_j| \geq D\left( \rho - \sum_{i=1} ^N \muxi,\rho - \sum_{i=1} ^N \muxi \right) - D(\rho,\rho) + 2 \sum_{i=1} ^N D(\rho,\muxi)- \sum_{i=1} ^N D(\muxi,\muxi). 
\end{equation}
\end{lem}

\begin{proof}
By Newton's theorem, Lemma \ref{lem:Newton}, the radiality of $\mu$ implies
\[
- \sum_{i\neq j} \log |x_i-x_j| \geq \sum_{i\neq j} D(\muxi,\muxj) 
\]
with (by the way) equality if $\mu$ is supported in the unit disc and $\min_{i,j} |x_i-x_j|\leq l$. The rhs of the above is equal to the rhs of \eqref{eq:Onsager}, which proves the lemma. Note that
\[
D\left( \rho - \sum_{i=1} ^N \muxi,\rho - \sum_{i=1} ^N \muxi \right) \geq 0 
\]
using Lemma \ref{lem:Coul def pos}, because by assumption 
\[
 \intR \left( \rho -\sum_{i=1} ^N \muxi\right) = 0.
\]
\end{proof}

The first term in the right-hand side of \eqref{eq:Onsager} is usually dropped to obtain a convenient lower bound to the Coulomb Hamiltonian. The core of our argument consists in obtaining a bound on its expectation value in the Gibbs measure from our upper and lower bounds to the free energy and using it to control the fluctuations around the mean field density. As we shall prove, this allows to obtain estimates on the marginals of the Gibbs measure and in particular \eqref{eq:1 particle lim}.

The functions $\muxi$ in Lemma \ref{lem:Onsager} should be thought of as unit charges smeared over small balls that we use to replace the point charges of the Coulomb gas. This is essential in our approach but has some cost that we quantify in the next lemma, which is an adaptation of a well-known lemma used by Lieb and Oxford (cf \cite[Chapter 6]{LieSei} for references).

\begin{lem}[\textbf{The cost of smearing out charges}]\mbox{}\label{lem:Lieb}\\
Here $\mu$ denotes the normalized (in $L ^1$) characteristic function of the disc of radius $l$. For any $\rho \in L ^{\infty}(\R ^2)$
\begin{equation}\label{eq:Lieb}
\left|D\left( \rho, \delxi - \muxi \right) \right| \leq C l ^2 \left\Vert \rho \right\Vert_{L ^{\infty}} 
\end{equation}
\end{lem}

\begin{proof}
We denote $h_{\nu}$ the potential associated to a charge distribution $\nu$. By Newton's theorem
\[
h_{\muxi} = h_{\delxi} 
\]
in $\R ^2 \setminus B(0,l)$. We are thus left with computing
\[
 \int_{B(0,l)} \rho \left( h_{\muxi} - h_{\delxi} \right)
\]
which is easily found to be equal to
\[
\int_{B(0,l)} \rho \left( \frac{1}{2\pi} \log\frac{r}{l} - \frac{r^2 -l^2}{2l^2} \right) rdrd\theta.
\]
The second term is easily bounded by a $C l^2\left\Vert \rho \right\Vert_{L ^{\infty}}$ while the first term is proportional to (adapting the analogous computation in 3D, cf \cite{Lie} and \cite[Chapter 6]{LieSei})
\begin{equation}\label{eq:pouet}
\int_{0} ^l \frac{F(r)}{r} dr  
\end{equation}
with
\[
F(r):= \frac{1}{2\pi} \int_{0} ^{2\pi} \int_{0} ^r \rho(s,\theta) sdsd\theta.  
\]
Using 
\[
 |F(r)|\leq \left\Vert \rho \right\Vert_{L ^{\infty}} \frac{r^2}{2},
\]
we can bound \eqref{eq:pouet} and complete the proof of the lemma.
\end{proof}

We can now proceed to the 

\noindent \emph{Proof of Theorem \ref{teo:MF limit}}

The upper bound is proved by taking the trial state $(\rhoMF ) ^{\otimes N}$ in \eqref{eq:free ener f}, \eqref{eq:free ener e}. The $- D(\rhoMF,\rhoMF)$ error term comes from the fact that there are $N(N-1)$ pairs of particles, to be divided by the $1/N$ mean-field scaling factor.

\medskip

For the lower bound we use Onsager's lemma with $\rho = N\rhoMF$, $\mu$ the normalized indicative function of the unit ball and $l = N ^{-1/2}$ to obtain
\begin{multline}\label{eq:lowbound HN}
H_N \left(x_1,\ldots,x_N \right) \geq \sum_{i=1} ^N \Vm (x_i) + \frac{2}{N} D\left( N\rhoMF - \sum_{i=1} ^N \muxi, N\rhoMF - \sum_{i=1} ^N \muxi \right) \\ 
-2N D(\rhoMF,\rhoMF) +  4  \sum_{i=1} ^N D(\rhoMF,\muxi)- 2 \sum_{i=1} ^N\frac{1}{N} D(\muxi,\muxi)
\end{multline}
and we will drop the second term of the right-hand side, which is positive, for the moment. 

We now invoke Lemma \ref{lem:Lieb} and use \eqref{eq:bounds rhoMF} to claim that
\[
D(\rhoMF,\muxi) = D(\rhoMF,\delxi) + O(N ^{-1}) = \hMF(x_i) + O(N ^{-1})
\]
which turns into 
\begin{equation}\label{eq:proof1}
4 D(\rhoMF,\muxi) = \MFe + 2 D(\rhoMF,\rhoMF) - \Vm (x_i) - T \log \rhoMF (x_i)  + O(N ^{-1}) 
\end{equation}
thanks to the variational equation \eqref{eq:EEL MFm}. Inserting into \eqref{eq:lowbound HN} we obtain
\begin{eqnarray}\label{eq:lowbound HN 2}
 H_N \left(x_1,\ldots,x_N \right) &\geq& N \MFe - T \sum_{i=1} ^N \log \rhoMF (x_i) -\frac{2}{N} \sum_{i=1} ^N D(\muxi,\muxi) - C \nonumber \\
 &=& N \MFe - T \sum_{i=1} ^N \log \rhoMF (x_i) - D(\mu,\mu) + \log \left( \frac{1}{\sqrt{N}} \right) -C
\end{eqnarray}
where the second line follows from a simple computation (recall that $\muxi$ is a unit charge smeared over the ball $B(x_i,N ^{-1/2})$). There only remains to use \eqref{eq:lowbound HN 2} to compute a lower bound to the free energy of $\muN$ : 
\begin{eqnarray*}
\FN[\muN] &\geq&  N \MFe - C - \frac{1}{2} \log N + T \intRN \muN \left( \log \muN - \sum_{i=1} ^N \log \rhoMF (x_i) \right)dx_1 \ldots dx_N \\
&=& N \MFe - C - \frac{1}{2} \log N + T \intRN \muN  \log \left( \frac{\muN}{\rhoMF \mbox{} ^{\otimes N}} \right)                                                                                                                                                  
\end{eqnarray*}
which ends the proof of the lower bound since the last term is the relative entropy of $\muN$ with respect to $(\rhoMF) ^{\otimes N}$ and hence is positive by Lemma \ref{lem:rel ent}.

\medskip 

Going back to \eqref{eq:lowbound HN} and retrieving the positive term we had discarded we see that our upper and lower bounds to the energy imply the additional estimate 
\begin{equation}\label{eq:factor rigor}
\intRN \muN (x_1,\ldots,x_N) D\left( N\rhoMF - \sum_{i=1} ^N \muxi, N\rhoMF - \sum_{i=1} ^N \muxi \right) dx_1\ldots dx_N \leq C N \log N .
\end{equation}
This bound quantifies how close $\muN$ is to $(\rhoMF ) ^{\otimes N}$. In particular it implies \eqref{eq:1 particle lim}, a fact that we state as a lemma:

\begin{lem}[\textbf{Using the Onsager term}]\label{lem:key}\mbox{}\\
For any $V : \R ^2 \mapsto \R$ regular enough and any symmetric probability measure $\muN \in \PP_s ((\R ^{2})^N)$ 
\begin{multline}\label{eq:key}
\left| \intR V(z) \left( \muN ^{(1)} (z)-\rhoMF (z)\right) dz\right| \leq C N ^{-1/2} \left\Vert \nabla V \right\Vert_{L ^{\infty}} \\ + C N ^{-1}  \left\Vert \nabla V \right\Vert_{L ^{2}} \left(\intRN \muN (Z) \: D\left( N \rhoMF - \sum_{i=1} ^N \muxi,N\rhoMF - \sum_{i=1} ^N \muxi \right) \right) ^{1/2}.
\end{multline}
\end{lem}

\begin{proof}

Given a one-body potential $V$ we start with
\begin{equation}\label{eq:proof3}
 \intR V \left(\muNone-\rhoMF\right) = \frac{1}{N} \intRN \muN (x_1,\ldots,x_N) \left(\sum_{i=1} ^N V(x_i) - N \intR V \rhoMF \right) dx_1 \ldots dx_N,
\end{equation}
which follows from the symmetry and normalization of $\muN$. Next, note that (with $\muxi$ defined as above)
\begin{equation}\label{eq:proof2}
\sum_{i=1} ^N V(x_i) = \intR V \sum_{i=1} ^N \delxi = \intR V \sum_{i=1} ^N  \muxi + O\left(\sqrt{N}\left\Vert \nabla V \right\Vert_{L^{\infty}}\right) 
\end{equation}
where we use that $\muxi$ is a unit charge smeared over a ball of radius $N ^{-1/2}$. Thus
\begin{multline}\label{eq:proof add}
 \intR V \left(\muNone-\rhoMF\right) =  \frac{1}{N} \intRN \muN (x_1,\ldots,x_N) \left(\intR V\left( \sum_{i=1} ^N \muxi - N \rhoMF\right) \right) dx_1 \ldots dx_N \\+ O\left( N ^{-1/2} \left\Vert \nabla V \right\Vert_{L^{\infty}}\right).
\end{multline}
Then
\begin{align}
\left| \intR V \left( N \rhoMF - \sum_{i=1} ^N \muxi \right) \right|  &= \left| \intR \nabla V \cdot \nabla h \right|\nonumber
\\ &\leq C \left\Vert \nabla V \right\Vert_{L ^2 (\R ^2)} \left(\intR \left| \nabla h \right| ^2  \right)^{1/2} \nonumber
\\&\leq C \left\Vert \nabla V \right\Vert_{L ^2 (\R ^2)} D \left(N \rhoMF - \sum_{i=1} ^N \muxi ,N \rhoMF - \sum_{i=1} ^N \muxi  \right) ^{1/2} \label{eq:calcul Onsager}
\end{align}
where 
\[
h = h_{N \rhoMF - \sum_{i=1} ^N \muxi} = 2\pi \left( -\Delta \right)^{-1} \left( N \rhoMF - \sum_{i=1} ^N \muxi \right). 
\]
To justify these computations we argue as in Step 1 of the proof of Proposition \ref{pro:decay MF}, using that the Coulomb kernel is the Green function of $-\Delta$ in $\R ^2$ and that 
\[
\intR N \left(\rhoMF - \sum_{i=1} ^N \muxi\right) = 0. 
\]

On the other hand %\footnote{$dX = dx_1\ldots dx_N$} 
\begin{multline*}
\left|\intRN dX \muN (x_1,\ldots,x_N) \intR V(z) \left( N\rhoMF (z)- \sum_{i=1} ^N \muxi(z)\right)dz \right| \\ 
\leq \left( \intRN dX \muN (x_1,\ldots,x_N) \left( \intR V \left( N\rhoMF (z)- \sum_{i=1} ^N \muxi(z)\right)dz\right) ^2 \right) ^{1/2}
\end{multline*}
by the Cauchy-Schwarz inequality and using the normalization of $\muN$. Using \eqref{eq:calcul Onsager}, we thus have 
\begin{multline*}
\left|\intRN dX \muN (x_1,\ldots,x_N) \intR V(z) \left( N\rhoMF (z)- \sum_{i=1} ^N \muxi (z)\right)dz \right| \\ 
\leq C \left\Vert \nabla V \right\Vert_{L ^2 (\R ^2)} \left( \intRN \muN  D \left(N \rhoMF - \sum_{i=1} ^N \muxi ,N \rhoMF - \sum_{i=1} ^N \muxi  \right) \right) ^{1/2} 
\end{multline*}
that we insert into \eqref{eq:proof add} to conclude the proof.
\end{proof}

Our final estimate \eqref{eq:1 particle lim} follows by combining \eqref{eq:factor rigor} and \eqref{eq:key}.

\hfill \qed

\begin{rem}[Estimates for higher-order marginals]\label{rem:marginals}\mbox{}\\ 
Our approach can give estimates on not only the first marginal $\muNone$ but on any reduced density $\muN ^{(k)}$, provided $k$ is suitably small as compared to $N$ (for example, $k$ fixed when $N\to \infty$). This can be seen to be a consequence of our main technical estimate \eqref{eq:factor rigor} as follows: Consider for example a smooth $2$-body potential $V_2(z_1,z_2)$ and evaluate
\begin{align*}
 \int_{\R ^4} \muN ^{(2)} V_2 &= \frac{1}{N^2} \sum_{1\leq i, j \leq N} \intRN \muN (x_1,\ldots,x_N) V_2(x_i,x_j) dX\\
 &= \frac{1}{N^2} \sum_{1\leq i, j \leq N} \intRN \muN (x_1,\ldots,x_N) \left(\int_{\R ^4} V_2 (z_1,z_2) \muxi (z_1) \muxj (z_2) dz_1 dz_2 \right)dX \\ &+ O\left(N^{-1/2} \sup_{y} \left\Vert \nabla V_2 (y,.) \right\Vert_{L^{\infty} }\right) + O\left(N^{-1/2} \sup_{y} \left\Vert \nabla V_2 (.,y) \right\Vert_{L^{\infty} }\right)
\end{align*}
as in \eqref{eq:proof2}. Then 
\begin{multline*}
\left| \sum_{i,j} \int_{\R ^4} V_2 (z_1,z_2) \muxi (z_1) \muxj (z_2) dz_1 dz_2 - N ^2 \int_{\R ^4} V_2 (z_1,z_2) \rhoMF (z_1) \rhoMF (z_2) dz_1 dz_2\right| \\
= \left|  \int_{\R ^4} V_2 (z_1,z_2) \left(\left(\sum_{i=1} ^N \muxi(z_1)\right) \left(\sum_{i=1} ^N \muxi(z_2)\right)-N^2 \rhoMF (z_1) \rhoMF (z_2)\right)\right|
\\
= \left|  \int_{\R ^4} V_2 (z_1,z_2) \left(\left(\sum_{i=1} ^N \muxi(z_1)\right) \left( \sum_{i=1} ^N \muxi(z_2) - N\rhoMF (z_2)\right) + N\rhoMF(z_2)\left( \sum_{i=1} ^N \muxi(z_1) - N\rhoMF (z_1)\right) \right)\right|
\\
\leq C N \left( \sup_{y} \left\Vert \nabla V_2 (y,.) \right\Vert_{L^2 } + \sup_{y} \left\Vert \nabla V_2 (.,y) \right\Vert_{L^{2} }\right) D\left( N \rhoMF - \sum_{i=1} ^N \muxi, N \rhoMF - \sum_{i=1} ^N \muxi \right) ^{1/2}.
\end{multline*}
We have used the fact that $N ^{-1}\sum_{i=1} ^N \muxi$ and $\rhoMF$ are normalized in $L ^1$, integrated over $z_1$ and $z_2$ separately and argued as in \eqref{eq:calcul Onsager}. There only remains to use the Cauchy-Schwarz inequality as before to obtain an estimate of the form 
\begin{multline*}
\left| \int_{\R ^4} V_2 \left( \muN ^{(2)} -(\rhoMF)^{\otimes 2} \right) \right| \leq C N ^{-1/2} (\log N ) ^{1/2}  \left( \sup_{y} \left\Vert \nabla V_2 (y,.) \right\Vert_{L^2 } + \sup_{y} \left\Vert \nabla V_2 (.,y) \right\Vert_{L^{2} }\right) \\ 
+ C N ^{-1/2} \left( \sup_{y} \left\Vert \nabla V_2 (y,.) \right\Vert_{L^{\infty} } + \sup_{y} \left\Vert \nabla V_2 (.,y) \right\Vert_{L^{\infty} }\right),
\end{multline*}
which generalizes \eqref{eq:1 particle lim}. Estimates for $\muN ^{(k)},k>2$ follow along the same lines. 
\end{rem}
\hfill \qed

\begin{rem}[More information when $T$ is larger]\mbox{}\\
Note that we have dropped one term that we could have estimated in the proof above, namely we have 
\begin{equation}\label{eq:estim bonus 1}
 \intRN \muN \log \frac{\muN}{\rhoMF \: ^{\otimes N}} \leq C T ^{-1} \log N.
\end{equation}
This estimate is useless in the regime $T= N^{-1}$ which interests us most but can become interesting when $T$ is larger, in particular in the somehow more natural case where $T$ is fixed. 

Indeed, using subadditivity of entropy (see e.g. \cite[Proposition 1]{Kie1})
\[
\intRN \muN \log \muN \geq N \intR \muNone \log \muNone 
\]
and the fact that 
\[
\intRN \muN \log (\rhoMF ) ^{\otimes N} = N \intR \muNone \log \rhoMF, 
\]
\eqref{eq:estim bonus 1} implies 
\[
\intR \muNone \log \frac{\muNone}{\rhoMF} \leq C \frac{\log N}{T N}. 
\]
This can be turned into 
\begin{equation}\label{eq:estim bonus 2}
\left\Vert \muNone - \rhoMF \right\Vert_{\rm TV} \leq C \sqrt{\frac{\log N}{TN}},
\end{equation}
where $\left\Vert \: . \: \right\Vert_{\rm TV}$ stands for the total variation norm, thanks to the Csisz\'{a}r-Kullback-Pinsker inequality recalled in Lemma \ref{lem:rel ent}. As already mentioned, \eqref{eq:estim bonus 2} becomes interesting only for relatively large temperature (take e.g. $T$ fixed independently of $N$), in which case it is a somehow better estimate than \eqref{eq:1 particle lim}.

Our approach to the mean-field limit in the thermal regime, presented in the next section, is based on this kind of considerations.

\end{rem}
\hfill\qed

\medskip

We now conclude the proof of Theorem \ref{teo:QH phases}, Item 1. Our main estimate \eqref{eq:1 particle lim el} follows by combining \eqref{eq:1 particle lim} and \eqref{eq:MF elec regime}. Indeed, using the Cauchy-Schwarz inequality in Fourier space as before, 
\[
\left| \intR \left(\rhoMF - \rhoMFel\right) V \right| \leq C D(\rhoMF - \rhoMFel, \rhoMF - \rhoMFel) ^{1/2} \left\Vert \nabla V \right\Vert_{L^2 (\R ^2)}.
\]

\medskip

We conclude this  section by giving the proof of \eqref{eq:1 particle decay}. We use the explicit expression \eqref{eq:plasma analogy} and the lower bound \eqref{eq:lowbound HN 2} to obtain
\[
\muN (z_1,\ldots,z_N) \leq \frac{1}{\ZN} \exp\left( \sum_{j=1} ^N \log \rhoMF (z_j) - \frac{N}{T} \MFe + \frac{\log N}{2 T} + \frac{C}{T} \right).  
\]
On the other hand, we recall that
\[
 \FNe = - T \log \ZN, 
\]
so, using \eqref{eq:up bound} we have 
\[
 \ZN \geq \exp \left( - \frac{N}{T} \MFe + \frac{C}{T} \right)
\]
and we conclude that
\[
\muN (z_1,\ldots,z_N) \leq \exp\left( \sum_{j=1} ^N \log \rhoMF (z_j) + C \frac{\log N}{T} \right) \leq \prod_{j=1} ^N \rhoMF(z_j) \exp(CN\log N)
\]
which implies \eqref{eq:1 particle decay} after an integration over $z_2,\ldots,z_N$ and the use of \eqref{eq:decay MF}.

\subsection{Thermodynamic limit in the thermal regime}\label{sec:MF thermal}

We now turn to the proof of \eqref{eq:1 particle lim th}. In contrast to what we did in the preceding section to prove \eqref{eq:1 particle lim el} we do not work in two steps, first relating $\muNone$ to $\rhoMF$ and then $\rhoMF$ to $\rhoMFth$. The reason is that the errors in the energy estimates produced by the use of Lemmas \ref{lem:Onsager} and \ref{lem:Lieb} are not sufficiently small compared to the difference between $\rhoMF$ and $\rhoMFth$ proved in \eqref{eq:MF ther regime coul}. One can however rely on a different strategy, treating the two-body term in \eqref{eq:free ener f} as a perturbation of the one-body part. Indeed, since we were able to prove that $\rhoMF$ is close to $\rhoMFth$ when $m\gg N ^2$, we already have an indication that the two-body Coulomb term is not so important in this regime. 

\medskip

\emph{Proof of \eqref{eq:1 particle lim th}.}

It is clearer to work in variables where the one-body potential takes its minimum when $r=1$, with value $0$. Scaling distances by a factor $(m/N) ^{1/2}$ it is equivalent\footnote{We do not change the notation for the scaled quantities.} to minimize the free-energy functional
\begin{equation}\label{eq:free ener rescal}
\FN [\mu] := \intRN \HN(Z) \mu(Z) dZ + \frac{1}{m} \intRN \mu(Z) \log \mu(Z) dZ
\end{equation}
with the rescaled Coulomb Hamiltonian
\begin{equation}\label{eq:Coul hami rescal}
\HN (Z) : =  \sum_{j=1} ^N \Vm (z_j) - \frac{2}{m} \sum_{i\neq j} \log |z_i - z_j|
\end{equation}
and
\begin{equation}\label{eq:Vm rescal}
\Vm (z) = |z| ^2 - 2 \log |z| -1.
\end{equation}
Note that we have taken advantage of the normalization of $\mu$ to subtract $1$ and have $\Vm(r) \geq \Vm(1) = 0$ for any $r\in \R$. 
%Let us introduce $h$ the potential generated by a unit charge spread over the circle of radius $1$:
%\begin{equation}\label{eq:h}
%h(r) :=  - (\log r ) \one_{r\geq 1}. 
%\end{equation}
%Our intuition is that $\rhoMFth$ is close to a delta function on the unit circle for large $m$, as follows from the results of Section \ref{sec:MF func}, and thus $h$ is a good approximation to $h_{\rhoMFth}$. 
We can rewrite the free energy as 
\begin{multline}\label{eq:free ener rewrite}
\FN [\mu] = N \int_{\R ^4}  \left( \frac{\Vm(x)}{2} + \frac{\Vm(y)}{2} - \frac{2(N-1)}{m} \log |x-y| \right) \mu ^{(2)}(x,y) dxdy 
\\ + \frac{1}{m} \intRN \mu(Z) \log \mu(Z) dZ 
\end{multline}
where $\mu^{(2)}$ is the two-body density of $\mu$. We write, for some parameter $\beta$ to be fixed later on, 
\begin{equation}\label{eq:W pot}
 \frac{\Vm(x)}{2} + \frac{\Vm(y)}{2} - \frac{2(N-1)}{m} \log |x-y| \geq \left( 1-\beta \right) \left( \frac{|x| ^2}{2} + \frac{|y| ^2}{2} \right) - \log |x| - \log |y| -1 +\Gamma_{\beta} 
\end{equation}
with 
\begin{equation}\label{eq:Gamma beta}
\Gamma_{\beta} = \inf_{x,y \in \R ^2} \left( \beta \left( \frac{|x| ^2}{2} + \frac{|y| ^2}{2} \right) - \frac{2(N-1)}{m} \log |x-y| \right).  
\end{equation}
Minimizing 
$$\varphi(x,y) := |x| ^2 + |y| ^2 - \delta \log |x-y|$$
with respect to $y$ we find that the minimum is attained for $|y| = \frac{x}{2}\left( 1 - \sqrt{1+2\delta / |x| ^2} \right)$ and that 
\begin{align*}
\varphi(x,y) &\geq \frac{3}{2}|x| ^2 - \frac{\delta}{2} - \frac12 |x| \sqrt{|x| ^2 + 2 \delta} - \delta \log \left( \frac{|x| + \sqrt{|x| ^2 + 2 \delta}}{2}\right)\\
&\geq |x| ^2 - \delta \log \left( \frac{|x| + \sqrt{|x| ^2 + 2 \delta}}{2}\right) \geq |x| ^2 - \delta \log\left(\sqrt{|x| ^2 + 2 \delta} \right). 
\end{align*}
Minimizing the last expression with respect to $x$ we obtain
\[
\varphi(x,y) \geq \frac94 \delta ^2 - \frac{\delta}{2} \log \left( \frac{\delta}{2}\right) 
\]
from which we deduce, taking $\delta = 4(N-1)/m\beta$,
$$\Gamma_{\beta} = \inf_{x,y\in \R ^2} \frac{\beta}{2} \varphi(x,y) \geq  \frac{\beta}{2} \left(36 \frac{(N-1)^2}{m ^2 \beta ^2} - 2\frac{N-1}{m\beta} \log \left( 2\frac{N-1}{m\beta} \right)\right).$$
We can now choose 
\[
 \beta = 2\frac{N-1}{m}
\]
and deduce 
\[
\Gamma_{\beta} \geq 0. 
\]
We can thus bound from below the free energy \eqref{eq:free ener rescal} as 
\begin{equation}\label{eq:low bound free ener}
\FN[\mu] \geq \FNt [\mu ] := \intRN \HNt (Z) \mu(Z) dZ + \frac{1}{m} \intRN \mu(Z) \log \mu(Z) dZ
\end{equation}
where 
\begin{equation}\label{eq:Hamiltonian prime}
\HNt (Z) = \sum_{j=1} ^N \left(\left( 1-\beta \right) |z_j| ^2 - 2 \log |z_j| - 1 \right)
\end{equation}
is now a one-body operator. The minimum $\FNet$ of $\FNt$ over symmetric probability measures $\mu$ is attained at 
\begin{equation}
\mu = (\rhoMFtth ) ^{\otimes N}  
\end{equation}
with 
\begin{equation}\label{eq:rho thermal prime}
\rhoMFtth (z)= \frac{1}{\Ztht} \exp\left( - m \left(\left( 1-\beta \right) |z| ^2 - 2 \log |z| - 1 \right) \right) 
\end{equation}
and 
\begin{equation}\label{eq:FNe prime}
 \FNet = - \frac{N}{m} \log (\Ztht). 
\end{equation}
Moreover we have for any symmetric probability measure $\mu$
\begin{align}
\FNt [\mu] &\geq \FNet + \frac{1}{m} \intRN \mu \log \frac{\mu}{(\rhoMFtth ) ^{\otimes N} }\\
&\geq \FNet + \frac{N}{m} \intR \mu ^{(1)} \log \frac{\mu ^{(1)}}{\rhoMFtth }\label{eq:low bound thermal},
\end{align}
using the subadditivity of the entropy (see e.g. \cite[Proposition 1]{Kie1}) for the second inequality. The equivalent of formula \eqref{eq:exp rhoMFmth} in rescaled coordinates reads 
\begin{equation}\label{eq:rho therm rescal}
\rhoMFth (r) = \frac{1}{\Zth} \exp \left( - \frac{1}{m} \left( r^2 - 2 \log r -1\right) \right), \quad \MFeth = -\frac{1}{m} \log \Zth  
\end{equation}
with $\Zth$ a normalization constant (we again keep the same notation for the quantities after the scaling of distances). In the new coordinates, $\rhoMFth$ resembles a Gaussian centered on the minimum of $\Vm$ at $r=1$ with characteristic length $m ^{-1/2}$. Its maximum is thus of order $m ^{1/2}$ and it decays exponentially fast in the region where $|r-1|\gg m ^{-1/2}$. Since $\beta = O(N/m)$ it is not difficult to realize, using \eqref{eq:rho thermal prime} and \eqref{eq:rho therm rescal}, that 
\[
\log \Zth = \log \Ztht  + O\left(\frac{N}{m}\right).
\]
and thus it follows from \eqref{eq:FNe prime} and \eqref{eq:low bound thermal} that 
\begin{equation}
\FNe \geq  N \MFeth + \frac{N}{m} \intR \mu ^{(1)} \log \frac{\mu ^{(1)}}{\rhoMFtth } + O(N^2 m ^{-2}).\label{eq:low bound thermal 2}
\end{equation}
We now use $(\rhoMFth ) ^{\otimes N}$ as a trial state for $\FN$ and obtain
\[
\FNe \leq N \MFeth + \frac{2N(N-1)}{m} D(\rhoMFth,\rhoMFth).
\]
Arguing as in the proof of Theorem \ref{pro:MF func}, Step 4, we approximate $\rhoMFth$ by a delta function along the circle of radius $1$ to obtain
\[
D(\rhoMFth,\rhoMFth) = \intR \rhoMFth h_{\rhoMFth} = h_{\rhoMFth} (1) + O(m^{-1/2}).
\]
But, using Newton's theorem \eqref{eq:newton} and the exponential decay of $\rhoMFth$ for $|r-1|\gg m ^{-1/2}$, 
\[
h_{\rhoMFth} (1) = -2\pi \int_{r\geq 1} \rhoMFth(r) (\log r) rdr = O(m ^{-1/2}), 
\]
since the integral is located in a region where $\log r = O(m^{-1/2})$ and $\rhoMFth$ is normalized. We thus have the upper bound
\begin{equation}\label{eq:up bound thermal proof}
\FNe \leq  N \MFeth + C N^2 m ^{-3/2}, 
\end{equation}
which, combined with \eqref{eq:low bound thermal 2} and the CKP inequality \eqref{eq:CKP} gives 
\[
\left\Vert \muNone  - \rhoMFtth  \right\Vert_{\rm TV} \leq C N^{1/2} m^{-1/4}  
\]
and there only remains to note that also 
\[
 \left\Vert \rhoMFth - \rhoMFtth  \right\Vert_{\rm TV} \leq C N^{1/2} m^{-1/4} 
\]
and scale variables back to deduce the desired result.

\hfill \qed 

\medskip

\noindent\emph{Proof of \eqref{eq:1 particle decay thermal}.} In the course of the proof above we have established (still in rescaled variables)
\[
 H_N (Z) \geq \sum_{j=1} ^N \left[\left( 1-\beta \right)|z_j| ^2 - 2 \log |z_j|\right]
\]
and 
\[
\FNe = - \frac{1}{m} \log \ZN = N\MFeth + O (N ^2 m ^{-3/2})
\]
where $\ZN$ is the normalization constant of $\muN$. In view of the expression \eqref{eq:rho thermal prime} of $\rhoMFtth $ we have 
\[
 \frac{1}{m} \log \rhoMFtth (z)= \MFetth  - \left( 1-\beta \right)|z| ^2 + 2 \log |z| +1
\]
and thus
\[
 \muN(z_1,\ldots,z_N) = \frac{1}{\ZN} \exp\left( - m H_N (Z)\right) \leq \prod_{j=1} ^N \rhoMFtth  (z_j) \exp(N^2 m^{1/2}).
\]
After integration over $N-1$ variables and inspection of the expression for $\rhoMFtth $ it follows that 
\[
 \muNone (z) \leq \exp(-c m (r-1) ^2)
\]
when $|r-1| \geq C N m ^{-3/4}$ for $C$ large enough. Then \eqref{eq:1 particle decay thermal} is obtained by a change of scales.
\hfill\qed

%It is not difficult to see that 
%\[
%D(\delta_{r=\ropt},\delta_{r=\ropt}) = \int_{r=\ropt} h = h(\ropt) + O(?).  
%\]
%The explicit expression \eqref{eq:exp rhoMFmth} and the estimate \eqref{eq:MF ther regime coul} on the other hand imply 
%\[
%D(\rhoMFth,\rhoMFth) =  D(\delta_{r=\ropt},\delta_{r=\ropt}) + O(?).
%\]
%We thus have

%Then, arguing as in Step 4 of the proof of Theorem \ref{pro:MF func}, we have
%\begin{equation*}
%\FNe' = -\frac{N}{m} \log (\Zth\:') = N \left( \MFeth + 2 \frac{N-1}{m} h(1) + O( N m ^{-3/2}) \right) 
%\end{equation*}
%where the remainder is due to the variation of $h$ in the region where $\rhoMFth$ is concentrated, whose thickness is of order $m^{-1/2}$. Following again arguments from Step 4 of the proof of Theorem \ref{pro:MF func} we also have
%\begin{equation*}
%= N \left( \MFeth + 2 \frac{N-1}{m} D(\rhoMFth, \rhoMFth) + O( N m ^{-3/2}) \right).
%\end{equation*}

\section{Improved energy bounds for the LLL problem}\label{sec:energy}

%The proofs of our main results follow from the energy estimates we are able to deduce from the plasma analogy. The evaluation of the energy of the trial states is presented in Section \ref{sec:up bounds}, where we also present our elementary lower bounds. We can then conclude the proofs of Theorems \ref{teo:momentum} and \eqref{teo:correl} in Section \ref{sec:proofs}.

%\subsection{Energy bounds}\label{sec:up bounds}

In this section we use the results of Section \ref{sec:QHphases} to compute the energy of our trial states. This allows to improve the upper bounds in Theorem \ref{teo:result ener}. More precisely, we prove

\begin{pro}[\textbf{Energy upper bounds}]\label{pro:up bound}\mbox{}\\
In the limit $\om,k\to 0$, $N\to \infty$ we have 
\begin{equation}\label{eq:ener up bound}
\LLLe \leq \om N ^2 (1+o(1)) + \frac{4}{3} k N ^3 (1+o(1))
\end{equation}
if $\om \geq -2kN$,
\begin{equation}\label{eq:ener up bound 2}
\LLLe \leq -\frac{N \om ^2}{4 k} + \frac{1}{3} k N ^3  \left( 1+ o(1) \right)
\end{equation}
if $\om \leq -2kN$ and $|\om| \ll k N ^{7/5} \log N$, and finally 
\begin{equation}\label{eq:ener up bound 3}
\LLLe \leq -\frac{N \om ^2}{4 k} - \frac{3}{2} \om N  \left( 1+ o(1) \right)
\end{equation}
if $\om \leq -2kN$ and $|\om| \gg k N ^{10/3}$.
\end{pro}

\begin{proof}
We use the trial states \eqref{eq:trial states}. Since they all have zero interaction energy, we are left with estimating (remember the scaling of space variables in \eqref{eq:Gibbs states})
\begin{align}\label{eq:rescal energy}
\LLLf[\PsiGV_m] &= \intRN  \sum_{j=1} ^N \pot(r_j) \left| \PsiGV_m (Z) \right| ^2 dZ \nonumber \\
&= N ^2 \left( \intR \left(\om r ^2 + k N r^4\right) \muNone (z) \right) dz.
\end{align}
Our main tools are \eqref{eq:1 particle lim el} and \eqref{eq:1 particle lim th} which essentially say that $\muNone\approx \rhoMFel$ or $\muNone \approx \rhoMFth$, depending on the regime. The main terms in the right-hand sides of \eqref{eq:ener up bound} and \eqref{eq:ener up bound 2} are obtained by replacing directly $\muNone$ by $\rhoMFel$ or $\rhoMFth$ and our main task is to estimate the error. We cannot use \eqref{eq:1 particle lim el} or \eqref{eq:1 particle lim th} directly however because the norms appearing in the right-hand sides are certainly not finite for $V = \om r ^2 + k N r^4$. We therefore first employ \eqref{eq:1 particle decay} (or \eqref{eq:1 particle decay thermal}) to restrict the integration domain : We will use two smooth radial cut-offs functions $\chiin$ and $\chiout$ satisfying 
\begin{equation}\label{eq:chi}
\chiin + \chiout = 1  
\end{equation}
and decompose $\LLLf[\PsiGV_m]$ as 
\begin{align}\label{eq:rescal energy 2}
\LLLf[\PsiGV_m] &= N^2  \intR \chiin \potN(z) \varrho (z) dz + N^2  \intR \chiin \potN(z) \left( \muNone (z) - \varrho(z) \right) dz \nonumber
\\&+ N^2 \intR \chiout \potN(z) \muNone(z) dz
\end{align}
where $\varrho = \rhoMFel$ or $\rhoMFth$ depending on the regime and we denote 
\begin{equation}\label{eq:potN}
\potN (r) = \om r^2 + kN r ^4. 
\end{equation}
In the electrostatic regime our choice of cut-offs functions will ensure $\chiin = 1$ on $\supp(\rhoMFel)$ and the first term is thus readily computed using the explicit expressions \eqref{eq:exp rhoMFmh 0} and \eqref{eq:exp rhoMFmh}~:
\begin{equation}\label{eq:term princ}
N^2 \intR \chiin \left(\om  r ^2 + k N  r^4\right) \rhoMFel (z) dz =   \om N ^2 \left( 1+ \frac{m}{N} \right) +  k N ^3 \left( \frac{4}{3} + 2 \frac{m}{N} + \frac{m ^2}{N ^2}\right).
\end{equation}
Optimizing the above expression with respect to $m$ we find 
\begin{equation}\label{eq:m opt}
\mopt = \begin{cases}
               0 \mbox{ if } \om \geq - 2 k N \\
               - N - \frac{\om}{2 k} \mbox{ if } \om < - 2 k N.
              \end{cases} 
\end{equation}
Therefore \emph{the Laughlin state is favored for $\om \geq - 2 k N$ whereas it is better to add a vortex at the origin for $\om < - 2 k N$.} The term \eqref{eq:term princ} becomes 
\begin{equation}\label{eq:term princ 1.5}
N^2 \intR \chiin \left(\om  r ^2 + k N  r^4\right) \rhoMFel (z) dz = - N \frac{\om ^2}{4k} + \frac{1}{3}kN ^3.
\end{equation}
if $\om \geq -2kN$, and 
\begin{equation}\label{eq:term princ 2}
N^2 \intR \chiin \left(\om  r ^2 + k N  r^4\right) \rhoMFel (z) dz = - N \frac{\om ^2}{4k} + \frac{1}{3}kN ^3
\end{equation}
if $\omega > -2kN$. In the thermal regime we keep the same expression for our choice of $m$ and obtain, using \eqref{eq:exp rhoMFmth}, 
\begin{equation}\label{eq:term princ thermal}
N^2 \intR \chiin \left(\om  r ^2 + k N  r^4\right) \rhoMFth (z) dz \leq -N\frac{\om ^2}{4k} - \frac{3}{2} \om N (1+o(1)).
\end{equation}
Note that the main terms above may be recovered from \eqref{eq:exp rhoMFmth} by neglecting terms beyond quadratic in a Taylor expansion of $\Vm$ around $\ropt$. 

When estimating the remainder terms in \eqref{eq:rescal energy 2} we distinguish between three regimes. The ``cases'' below refer to the different cases in Theorem \ref{teo:result ener}.

\medskip

\noindent \emph{Case 1 and 2, $\mopt \leq CN$.} We take $\varrho = \rhoMFel$ in \eqref{eq:rescal energy 2}. The support of $\rhoMFel$ is uniformly bounded in this regime. Using \eqref{eq:1 particle decay} we see that for $r\geq \Rplus$
\begin{equation}\label{eq:estim decay 1}
 \muNone (r) \leq C \exp \left( -C N (r^2 - C \log N) \right)
\end{equation}
and thus $\muNone$ is exponentially small, both as a function of $r$ and $N$ for $r\geq C (\log N) ^{1/2}$. We choose
\[
 \chiin = 1 \mbox{ in } B(0,C (\log N) ^{1/2}), \quad \chiin = 0 \mbox{ out of } B(0,2C(\log N) ^{1/2})
\]
and $\chiout$ accordingly. Then $\chiin \equiv 1$ on the support of $\rhoMFel$ as desired and \eqref{eq:term princ} yields the main terms in the right-hand sides of \eqref{eq:ener up bound} and \eqref{eq:ener up bound 2}. We also assume $|\nabla \chiin| \leq C (\log N) ^{-1/2}$. We then note that 
\[
\left| \nabla (\chiin  \potN ) \right| \leq C (\log N) ^{1/2} |\om| + kN (\log N) ^{3/2} 
\]
for $r\leq C (\log N) ^{1/2}$ while
\[
\left\Vert  \chiin  \nabla \potN \right\Vert_{L ^2} \leq C (\log N) |\om| + kN (\log N) ^{2}. 
\]
Thus, using \eqref{eq:1 particle lim el} to estimate the second term in \eqref{eq:rescal energy 2} we obtain that it is bounded above by 
\[
C  \left(\om N ^{3/2}(\log N) ^{3/2}  + kN ^{5/2} (\log N) ^{5/2} \right)
\]
and \eqref{eq:ener up bound} follows, using \eqref{eq:term princ 1.5} or \eqref{eq:term princ 2} for the main term of \eqref{eq:rescal energy 2}. The third term in \eqref{eq:rescal energy 2} is negligible thanks to \eqref{eq:estim decay 1}.

\medskip

\noindent \emph{Case 3, $N \ll \mopt \ll N ^{7/5} \log N$.} We take again $\varrho = \rhoMFel$ in \eqref{eq:rescal energy 2}. In this case \eqref{eq:1 particle decay} gives
\[
 \muNone (r) \leq C \exp \left( -C N ((r-\ropt)^2 - C \log N) \right)
\]
for $r\geq \Rplus$ or $r\leq \Rminus$ and thus $\muNone$ is exponentially small for $|r-\ropt|\geq C (\log N) ^{1/2}$. We take 
\[
 \chiin = 1 \mbox{ for } \ropt - C (\log N)^{1/2} \leq r\leq \ropt + C (\log N)^{1/2}, \quad \chiin = 0 \mbox{ for } |r-\ropt|\geq 2C (\log N)^{1/2}
\]
and $\chiout = 1 -\chiin$ accordingly. As before the main terms of \eqref{eq:ener up bound 2} come from \eqref{eq:term princ} and the third term in \eqref{eq:rescal energy 2} can be neglected due to the exponential decay of $\muNone$. For $|r-\ropt|\leq C (\log N) ^{1/2}$ one can easily realize that
\[
\left| \nabla (\chiin  \potN ) \right| \leq C (\log N) ^{1/2} |\om|, 
\]
whereas 
\[
\left\Vert  \chiin  \nabla \potN \right\Vert_{L ^2} \leq C N ^{-1/4} (\log N) ^{3/4} |\om| ^{5/4} k ^{-1/4}. 
\]
Using \eqref{eq:1 particle lim el}, one can see that the error due to the second term in \eqref{eq:rescal energy 2} is bounded by 
\[
C N ^{5/4} (\log N) ^{5/4} |\om| ^{5/4} k ^{-1/4} + C |\om| N^{3/2} (\log N) ^{1/2} 
\]
and this is negligible in front of $kN ^3$ if $|\om| \ll k N ^{7/5} \log N$, which concludes the proof.

\medskip

\noindent \emph{Case 4, $\mopt \gg N ^{10/3}$.} Here we take $\varrho = \rhoMFth$ in \eqref{eq:rescal energy 2} and notice that \eqref{eq:1 particle decay thermal} ensures that $\muNone$ is exponentially small when
\[
 |r-\ropt| \gg L := \max (N ^{-1/2}, N ^{1/2} m ^{-1/4}).
\]
We therefore choose
\[
 \chiin (r) =\begin{cases} 1 \mbox{ if } |r-\ropt|\leq L \\
              0 \mbox{ if } |r-\ropt| \geq 2L.
             \end{cases}
\]
With this choice and using \eqref{eq:1 particle decay thermal}, we can neglect the third term in \eqref{eq:rescal energy 2} since it will be much smaller than the two others. The main term is computed using \eqref{eq:term princ thermal}. For the second term of \eqref{eq:rescal energy 2} we use \eqref{eq:1 particle lim th} to see that
\[
\left| N^2  \intR \chiin \potN(z) \left( \muNone (z) - \varrho(z) \right) dz \right|\leq  C N ^{5/2} k ^{1/4} \om ^{3/4} L^2
\]
where we approximate $\potN$ by its second variation around its minimum. The last quantity is 
$$O(N^{7/2} k ^{3/4} |\om| ^{1/4})$$ 
when $|\om|/k \propto m\ll N ^4$ (recall the choice of $m=\mopt$ in \eqref{eq:m opt}) and 
$$O(N ^{3/2} k ^{1/4} \om ^{3/4})$$
otherwise. We conclude that the error term is negligible in front of the subleading term $-\frac{3}{2} \om N$ in \eqref{eq:ener up bound 3} when $|\om| \gg k N ^{10/3}$, which concludes the proof.

\end{proof}

For completeness we reproduce the argument from \cite{RSY} allowing to see that at least the order of magnitude of our energy upper bounds is correct. This is rather simple: We consider the problem of minimizing the potential energy in a sector of given angular momentum. This  makes sense because we are in the lowest Landau level and we can rewrite the potential energy as in \eqref{eq:Bargh}. Let us denote 
\begin{equation}\label{eq:ener pot L}
E_0 (L) = \inf\left \{ \bral F, \sum_{j=1} ^N h_j F \ketr_{\BargN},\: \LL_N F = L F  ,\: \bral F, F \ketr_{\BargN} = 1\right\} 
\end{equation}
where 
\begin{equation}\label{eq:one part op}
h = (\om + 3 k) z \dd_z + k (z \dd_z) ^2  
\end{equation}
and $h_j$ is the same operator acting on the variable $z_j$.

We have the following lemma :
\begin{lem}[\textbf{Potential energy at given momentum}]\label{lem:low bound}\mbox{}\\
For any $L\in \N$  
\begin{equation}\label{eq:low bound ener L}
E_0 (L) \geq (\om + 3 k) L + k \frac{L ^2}{N}:=e(L) 
\end{equation}
with equality if $L$ is a multiple of $N$.
\end{lem}

\begin{proof}
Since $\LL_N$ and $\sum_{j=1} ^N L_j ^2$ commute, they can be diagonalized simultaneously and we have the operator inequality
\[
\sum_{j=1} ^N L_j ^2 \geq \frac{1}{N} \left(\sum_{j=1} ^N L_j\right) ^2, 
\]
from which \eqref{eq:low bound ener L} follows. The fact that there is equality when $L$ is a multiple of $N$ is proved by taking the trial state $(\pi \frac{L}{N} !) ^{-N/2} \prod_{j=1} ^N z_j ^{L/N}$.
\end{proof}

To see that the lower bounds of Theorem \ref{teo:result ener} follow, simply use this lemma with $L=L_0$, the momentum of a ground state, as estimated in Theorem \ref{teo:momentum}.

\bigskip

\textbf{Acknowledgments.}\\
NR thanks Xavier Blanc and Mathieu Lewin for helpful discussions during the early stages of this project. SS is supported by an EURYI award. Funding from the CNRS in the form of a PEPS-PTI project is also acknowledged. JY thanks the Institute Mittag Leffler for hospitality during his stay in the fall of 2012.

\end{document}